\documentclass[12pt]{article}
\usepackage{times}
\usepackage{amsfonts}
\usepackage{amsmath}
\usepackage{enumitem}
\usepackage{mathtools}
\usepackage{centernot}
\usepackage{soul}
\usepackage{booktabs}
\usepackage{latexsym}
\usepackage[ruled]{algorithm2e}
\SetAlFnt{\small}
\SetAlCapFnt{\small}
\SetAlCapNameFnt{\small}
\SetAlCapHSkip{0pt}
\IncMargin{-\parindent}
\usepackage{algorithmic}
\usepackage{nicefrac}
\usepackage[
  bookmarks=true,
  bookmarksnumbered=true,
  bookmarksopen=true,
  pdfborder={0 0 0},
  breaklinks=true,
  colorlinks=true,
  linkcolor=black,
  citecolor=black,
  filecolor=black,
  urlcolor=black,
]{hyperref}
\usepackage{graphicx}
\usepackage{xcolor}
\usepackage[margin=1in]{geometry}
\usepackage{amssymb}
\usepackage{amsmath}
\usepackage[square]{natbib}
\setcitestyle{numbers}
\usepackage{amsthm}
\usepackage[capitalize]{cleveref}
\usepackage{bbm}
\usepackage[normalem]{ulem}
\usepackage{tikz}
\usetikzlibrary{decorations.pathreplacing,angles,quotes}
\usepackage{MnSymbol}
\theoremstyle{plain}
\newtheorem{theorem}{Theorem}[section]
\newtheorem{corollary}[theorem]{Corollary}

\newtheorem{lemma}[theorem]{Lemma}
\newtheorem{proposition}[theorem]{Proposition}
\newtheorem{observation}[theorem]{Observation}
\newtheorem{remark}[theorem]{Remark}

\newtheorem{definition}[theorem]{Definition}

\newtheorem*{question*}{Open Question}
\DeclarePairedDelimiter\parentheses{(}{)}
\DeclarePairedDelimiter\braces{\{}{\}}
\DeclarePairedDelimiter\brackets{[}{]}
\DeclarePairedDelimiter\absolute{|}{|}

\newenvironment{lp*}{\begin{equation*}  \begin{array}{l>{\displaystyle}l>{\displaystyle}l}}{\end{array}\end{equation*}}

\newcommand{\supp}{{\mathrm{supp}}}
\newcommand{\poly}{{\mathrm{poly}}}

\begin{document}
\title{Algorithmic Cheap Talk}
\date{September 9, 2024}
\author{Yakov Babichenko\thanks{Technion--Israel Institute of Technology | \emph{E-mail}: \href{mailto:yakovbab@technion.ac.il}{yakovbab@technion.ac.il}.} 
\and 
Inbal Talgam-Cohen\thanks{Tel Aviv University and Technion--Israel Institute of Technology | \emph{E-mail}: \href{mailto:inbaltalgam@gmail.com}{inbaltalgam@gmail.com}.}
\and 
Haifeng Xu \thanks{University of Chicago | \emph{E-mail}: \href{mailto:haifengxu@uchicago.edu}{haifengxu@uchicago.edu}.} \and Konstantin Zabarnyi\thanks{Technion--Israel Institute of Technology | \emph{E-mail}: \href{mailto:konstzab@gmail.com}{konstzab@gmail.com}.}}
\maketitle
\begin{abstract}
The literature on strategic communication originated with the influential \emph{cheap talk} model, which precedes the Bayesian persuasion model by three decades. This model describes an interaction between two agents: \emph{sender} and \emph{receiver}. The sender knows some \emph{state} of the world which the receiver does not know, and tries to influence the receiver's action by communicating a cheap talk message to the receiver.

This paper initiates the systematic algorithmic study of cheap talk in a finite environment (i.e., a finite number of states and receiver's possible actions). We first prove that approximating the sender-optimal or the welfare-maximizing cheap talk equilibrium up to a certain additive constant or multiplicative factor is NP-hard. We further prove that deciding whether there exists an equilibrium in which the receiver gets utility higher than the trivial utility he can guarantee is NP-hard. Fortunately, we identify two naturally-restricted cases that admit efficient algorithms for finding a sender-optimal equilibrium -- a constant number of states or a receiver having only two actions.
\end{abstract}

\section{Introduction}
\label{sec:intro}
What information should an informed and rational player (the \emph{sender}) communicate to a decision maker (the \emph{receiver}) so as to optimally influence the decision? Starting with the seminal 1982 work of Crawford and Sobel~\cite{crawford1982strategic} (see also~\cite{milgrom1981good,grossman1981informational} and~\cite{green1980two} who introduced similar ideas in an unpublished paper before~\cite{crawford1982strategic}), this question has been extensively studied in economics for four decades. The two leading branches of the economic literature on strategic information transmission are \emph{cheap talk} and \emph{Bayesian persuasion}. The fundamental difference between these two branches lies in the ability or disability of the sender to commit. The cheap talk model assumes no commitment power on the sender's side, while 
in the Bayesian persuasion model, the sender can commit to a strategically chosen information revelation policy in a trustworthy way.

The work of~\cite{dughmi2016algorithmic} initiated the study of algorithmic aspects of Bayesian persuasion, including the complexity of computing the sender-optimal signaling policy. Since then, Bayesian persuasion and its various extensions have been extensively studied in the algorithmic game theory (AGT) and theory of computation (ToC) literatures (see, e.g.,~\cite{babichenko2017algorithmic,dughmi2017algorithmic2,arieli2019private,cummings2020algorithmic,castiglioni2020online,xu2020tractability,GradwohlHHS21,hoefer2021algorithmic,castiglioni2021multi,fair-banerjee} and an algorithmic survey in \cite{dughmi2017algorithmic}). Additionally, signaling policies with commitment have been addressed in many contexts of the AGT literature, including auctions  \cite{DughmiIR14,cheng2015mixture,DaskalakisPT16,badanidiyuru2018targeting,alijani2022limits}), routing games \cite{bhaskar2016hardness,zhou2022algorithmic,griesbach2022public},  abstract games  \cite{dughmi2014hardness,rubinstein2015eth,bhaskar2016hardness} and recently in optimal stopping problems \cite{DingFHTX23}.

The complementary model of cheap talk is arguably more natural for strategic information communication; indeed, the economic literature mainly focused on cheap talk for three decades (see, e.g.,~\cite{forges1984note,Farrell87,matthews1989pre,austen1990information,forges1990equilibria,FarrellR96,crawford1998survey,Battaglini2002,AumannH2003,OttavianiS2006,Krishna2016}), before Bayesian persuasion was introduced in 2011 by Kamenica and Gentzkow~\cite{kamenica2011bayesian}. This early focus on cheap talk for modeling strategic communication stems from its minimalistic assumptions -- in many scenarios, it is unreasonable to assume that the sender credibly commits to a signaling policy as in Bayesian persuasion, whereas ``ordinary, informal talk’’~\cite{FarrellR96} more accurately describes the communication. \citet{lipnowski2020cheap}~give many examples for relevant scenarios, such as a political think tank advising a lawmaker or a broker advising an investor. The cheap talk model addresses such ordinary communication mathematically, via equilibrium concepts.

Despite being well-established and more realistic in many scenarios, to the best of our knowledge, cheap talk has not yet been systematically approached from an algorithmic perspective. \emph{The goal of this paper is to investigate the algorithmic aspects of cheap talk in its most basic form}. 

\subsection{Our Model and Results}
\label{sub:res_tec}
We study the classic cheap talk model in its most basic form with discrete states and actions. In this model, a \emph{state} of nature is drawn from a publicly known \emph{prior} distribution with a finite support. The setting includes two players: a \emph{sender} who eventually observes the state and a \emph{receiver} who does not observe it. For clarity of presentation, we use the pronouns she/her for the sender and he/his for the receiver. The sender can ``talk'' to the receiver to convey information ``cheaply". That is, the communication is without any cost, and the conveyed information -- a.k.a.~\emph{signal} -- can take any form, and it does not have to be accurate or correct. The receiver may then use or ignore the signal when choosing an action from a finite set of possibilities. The action coupled with the state determines the utilities of both the sender and the receiver.

As is usual in equilibrium analysis, we assume the sender is playing according to a strategy, known in our context as a \emph{signaling policy}. The signaling policy specifies the probability of sending different signals conditional on every possible state. Upon learning the true state, the sender draws a signal according to the marginal distribution induced by the signaling policy and sends it to the receiver. The receiver knows the prior distribution over the states, as well as the sender's signaling policy. After getting the signal, he applies a Bayesian update to deduce his \emph{posterior} belief (distribution) over the states. Based on this updated belief, he then chooses his (possibly random) action according to his own strategy. 

The two players' strategies form a \emph{cheap talk equilibrium}, or more formally a \emph{perfect Bayesian equilibrium}~\cite{fudenberg1991perfect}, if no agent has an incentive to deviate. Unlike general Nash equilibrium, finding a cheap talk equilibrium turns out to be a trivial task, since revealing no information always forms a cheap talk equilibrium (a.k.a.~the \emph{babbling equilibrium}). A popular perspective adopted in the literature on strategic information transmission is the sender's point of view: The goal is to analyze the \emph{sender-optimal} equilibrium, i.e., the equilibrium at which the sender's expected utility is maximized~\cite{crawford1982strategic,kartik2007credulity,lipnowski2020cheap,lin2022credible}. Therefore, the sender-optimal equilibrium will also be the main focus of our computational study, except in Subsection~\ref{sub:receiver} in which we discuss the implications of our results from the receiver's perspective.

Our main negative result shows that it is NP-hard to approximate the sender's utility value at the sender-optimal cheap talk equilibrium up to a certain multiplicative or additive constant (see Theorem \ref{thm:additive}). This result stands in contrast to the Bayesian persuasion model, in which a sender-optimal equilibrium and the sender's corresponding utility can be computed efficiently by a simple linear program. This may serve as a complexity-theoretic explanation of the explosive interest in Bayesian persuasion in the recent economic literature after its introduction in~\cite{kamenica2011bayesian}.

Our computational hardness result has implications for the economic analysis of cheap talk. It indicates that a simple characterization for the sender's optimal expected utility in cheap talk equilibria is unlikely to exist -- or, at least, the characterization has to be complex enough to encode certain NP-hard problems. This again stands in contrast to the following fundamental results in strategic information transmission:\footnote{The concavification of a function $f$ is the pointwise-smallest concave function weakly above $f$; similarly, the quasi-concavification of $f$ is the pointwise-smallest quasi-concave function weakly above $f$.}
\begin{enumerate}
    \item In the Bayesian persuasion model, the maximal sender's expected utility is precisely characterized as the concavification of the sender's \emph{indirect utility function} -- i.e., sender's utility as a function of the posterior distribution; see~\cite{kamenica2011bayesian}.
    \item In the cheap talk model with the simplification of state-independent sender's utility, the sender's optimal expected utility at an equilibrium is precisely characterized as the quasi-concavification of the sender's indirect utility function; see~\cite{lipnowski2020cheap}.
\end{enumerate}

Towards our positive results, we first show the existence of a sender-optimal equilibrium in which the sender uses no more signals than the number of states (Proposition~\ref{pro:signals}). Besides being an elegant mathematical property of cheap talk equilibria, this result also implies the NP-membership of deciding whether the sender can ensure certain utility through cheap talk. 

It is known from~\cite{vohra2023bayesian} that it is computationally tractable to find a sender-optimal cheap talk equilibrium when the sender's utility is state-independent (see also~\cite{aybas2019persuasion} for a closed formula for the corresponding sender's utility); this result follows from the characterization by~\cite{lipnowski2020cheap} of the utility the sender can guarantee at a cheap talk equilibrium in terms of the quasi-concavification of the sender's utility function. Our positive results identify further interesting tractable special cases of cheap talk. We investigate the source for the hardness of cheap talk -- is the problem hard due to a large number of states and/or due to a large number of actions? Our NP-hardness reduction uses cheap talk instances with many states and actions, and our positive results suggest that this hardness may reside in the combination of both. In particular, we show that for any constant number of states (and many actions), the sender-optimal equilibrium can be computed in polynomial time (Proposition~\ref{pro:states}). Also, for two actions (and many states), the sender-optimal equilibrium can be computed in linear time (Proposition~\ref{pro:binary}). We view this as a significant positive result since the \emph{binary-action} setting is a well-studied special case in the existing literature on strategic communication~\cite{li2016competitive,dughmi2017algorithmic2,kolotilin2017persuasion,guo2019interval,babichenko2022regret,feng2022rationality}. An interesting open question is whether there exists a polynomial-time algorithm for finding a sender-optimal cheap talk equilibrium with \emph{any constant} number of actions.

\subsection{Our Techniques}
\label{sub:tec}
There are two key differences between cheap talk and Bayesian persuasion that stem from the lack of the sender's commitment power in the former:

\begin{enumerate}
    \item Unlike Bayesian persuasion, in the cheap talk model the sender in each state must only transmit signals that are optimal in that state (i.e., signals that yield the best possible sender's utility given the receiver's strategy).
    \item In Bayesian persuasion, the sender can deviate to a strategy inducing a posterior distribution that is not induced by any signal under the original strategy. However, it will \emph{not} happen at any cheap talk equilibrium, as the receiver can  correctly interpret only the signals that are sent with \emph{strictly positive probability} according to the sender's strategy. To induce a new posterior, the sender has to use a new signal that the receiver is unfamiliar with, in which case his response can be arbitrary and hence potentially harmful to the sender.
\end{enumerate}

Each of the above differences is insufficient on its own to deduce a hardness result -- the linear programming approach that is useful in Bayesian persuasion can be adjusted to tackle it. The combination of \emph{both} these differences is the leading force of our hardness result. A key challenge is the continuous nature of (even the finite) cheap talk problem. Several techniques have been developed in the literature to reduce combinatorial problems to continuous game-theoretic ones. Unfortunately, we have not found the existing techniques applicable to cheap talk (for further discussion, see Subsection~\ref{sub:related} below). Instead, we develop a novel technique that uses a reduction from a variation of $3$SAT (see Section~\ref{sec:hard}).

\subsection{Related Literature}
\label{sub:related}
\textbf{Connections to designing AI agents for strategic communication.} Our theory work is of possible interest to more applied, cutting-edge research in AI: Driven by the recent maturity of generative language models, the AI/ML field is undergoing a shift of interest. From solving pure tactic-based games (such as Go~\cite{silver2017mastering}, Poker~\cite{brown2019superhuman}, or StarCraft~\cite{vinyals2019grandmaster}), the focus is shifting to games that require strategic communication via natural language, often in the style of cheap talk. Notable examples include a recent influential work on developing a human-level AI agent for the game of \emph{Diplomacy}, by combining language models with strategic reasoning~\cite{meta2022human}. \emph{Diplomacy} starts with a communication phase to exchange agents' private information -- in the style of cheap talk without any commitment --  followed by a decision-making phase that utilizes learned information from communication. Despite the short period since its publication, this work has spurred a rapidly-growing body of follow-up works on combining language communication with strategic reasoning, ranging from simple matrix games~\cite{gandhi2023strategic} to board games like \emph{Mafia}~\cite{o2023hoodwinked} or \emph{Werewolf}~\cite{xu2023language,xu2023exploring}, and even to complex real-world negotiation~\cite{abdelnabi2023llm}. In all these applications, the bare bones of the problem are a cheap talk game in which the sender tries to gain a strategic advantage by communicating information to another agent without commitment. An appropriate solution concept for such sequential Bayesian games is the perfect Bayesian equilibrium~\cite{fudenberg1991perfect} we study. We view our work as a potential complexity-theoretic contribution to this trend, since it studies the algorithmics of computing ``good'' perfect Bayesian equilibria in arguably the simplest possible situations.

\vspace{2mm}
\noindent \textbf{Complexity of equilibrium computation.} The complexity of equilibrium computation has been one of the core questions in algorithmic game theory. In a series of important works~\cite{lipton2003playing,daskalakis2009complexity,chen2009settling,rubinstein2016settling}, the complexity of finding a Nash equilibrium and the complexity of finding an approximate Nash equilibrium in a bi-matrix game was settled. From the computational perspective, the cheap talk game model we study has almost the same bi-matrix format, except that one of its dimensions represents the receiver's action, whereas another dimension represents the state, drawn from a known prior distribution. Thus, the sender's strategy cannot be distributed arbitrarily as in Nash equilibrium for a bi-matrix game. This subtle difference turns out to be fundamental --  finding a cheap talk equilibrium immediately becomes a trivial task since revealing no information is always an equilibrium.

\vspace{2mm}
\noindent \textbf{Complexity of ``optimal'' equilibrium computation.}
For Nash equilibria, it is known that the maximal clique problem~\cite{gilboa1989nash}, the Monotone-$1$-in-$3$SAT problem~\cite{deligkas2017computing} and the Boolean formula satisfaction problem~\cite{conitzer2003complexity} can be reduced to computing the ``optimal'' equilibrium. However, the problem of computing the optimal Nash equilibrium is significantly different from the sender-optimal cheap talk equilibrium. One difference is the \emph{non-local nature} of the sender's strategy, which is not present at Nash equilibrium. In particular, the hardness of finding the optimal Nash equilibrium in almost all the reductions above lies in identifying a subset of player actions -- often hinted by the solution to a maximal clique or the satisfying variable assignments -- so that mixing them leads to desirable equilibrium utility. Unlike such \emph{selection} of a subset of actions, the sender's signaling policy has to be a \emph{decomposition} of the prior distribution over all the states. We introduce a new variant of the $3$\texttt{SAT} problem and reduce to computing a sender-optimal cheap talk equilibrium from this variant, which better suits the decomposition interpretation of sender's~policies.

\vspace{2mm}
\noindent \textbf{Complexity of Bayesian persuasion.}
The hardness in signaling with commitment often comes from the difficulty of coordinating multiple receivers' actions, e.g., in zero-sum games~\cite{dughmi2014hardness,rubinstein2015eth} or routing~\cite{bhaskar2016hardness,zhou2022algorithmic}. The hardness in cheap talk, however, comes from a fundamentally different reason -- indeed, computing a sender-optimal cheap talk equilibrium is NP-hard even when there is a single receiver; in this case, Bayesian persuasion can be solved via linear programming.

\vspace{2mm}
\noindent \textbf{Cheap talk-related algorithms.} To the best of our knowledge, we are the first to systematically study computational aspects of the cheap talk model. Some recent works include algorithms related to cheap talk: The work of~\citet{lyu2022information} provides an algorithm that outputs the sender's optimal expected utility in a game combining cheap talk with Bayesian persuasion. Specifically, the sender in their model can commit to a signaling policy, but not to truthfully reporting the observed state as an input to the signaling policy. Unlike us, they do not focus on computational complexity. A concurrent work~\cite{condorelli2023cheap} studies cheap talk with the sender and the receiver being reinforcement learning agents with quadratic utilities; they empirically show that both agents' strategies converge to the unique Pareto efficient equilibrium.

\vspace{2mm}
\noindent \textbf{Other works on strategic communication.}  Besides the rich body of work on algorithmic Bayesian persuasion mentioned at the beginning, another related branch of literature on strategic information transmission studies the revelation of information that is \emph{ex-post verifiable}, prohibiting the sender from sending a false message. As a consequence, the sender cannot ``lie'' in her signaling policy; see, e.g.,~\cite{grossman1980disclosure,grossman1981informational,glazer2004optimal,HartKP17,titova2022persuasion}. The work of~\cite{HaghtalabIL+} studies a novel variant in which the sender wishes both to persuade and to inform. In their model, the sender communicates with ``anecdotes'' -- signals that are always truthful -- with and without commitment. Their focus is on equilibrium analysis rather than computation. \cite{hoefer2021algorithmic} also study persuasion using evident signals, both with and without commitment. They focus on a special case with binary receiver actions \emph{and} state-independent sender's utility. Interestingly, in their setup with constrained communications, optimal commitment is hard, while cheap talk is easy.

\subsection{Paper Organization}
\label{sub:organization}
In Section~\ref{sec:model}, we formally introduce cheap talk and provide a preliminary result that bounds the required number of signals in a sender-optimal cheap talk equilibrium; as a corollary, this implies that the decision version of computing the maximal sender's expected utility at an equilibrium belongs to NP. Section~\ref{sec:hard} shows that approximating a sender-optimal, or a welfare-maximizing, cheap talk equilibrium up to a certain multiplicative/additive constant is NP-hard. In subsection~\ref{sub:receiver}, we show that deciding whether the receiver can get equilibrium utility higher than at a babbling equilibrium is an NP-complete problem. Section~\ref{sec:tractble} analyzes two cases in which computing such an equilibrium is tractable -- a constant number of states and a binary-action receiver. Section~\ref{sec:discussion} summarizes, briefly discusses the computational complexity of social welfare maximization in cheap talk and mentions various directions for future research.

\section{The Model}
\label{sec:model}
\paragraph{Normal form cheap talk games.}
We study the cheap talk model in arguably its simplest representation, which we refer to as \emph{normal form}.\footnote{Many economic papers focus on cheap talk settings with continuum-sized state and action spaces. These settings are often structured; the utility functions are assumed to satisfy continuity and other regularity properties. See, e.g.,~\cite{crawford1982strategic,chen2008selecting}.} In a normal form cheap talk game $\langle A,\Omega,\mu,u_S,u_R \rangle$, there are two players: \emph{sender} (she) and \emph{receiver} (he). The receiver is a decision maker choosing which \emph{action} $a$ to take from a given finite set of $m$ actions $A=\braces*{a_1,\ldots,a_m}$. Both players' utilities are determined by the receiver's action, as well as a random \emph{state} of nature $\omega$, supported on a finite set $\Omega=\braces*{\omega_1,\ldots,\omega_n}$ of size $n$. We use $u_S\parentheses*{\omega, a}$ and $u_R\parentheses*{\omega, a}$ to denote the sender's and receiver's respective \emph{utilities}. The players share a common \emph{prior belief} $\mu=\parentheses*{\mu\parentheses*{\omega}}_{\omega\in\Omega}\in\Delta\parentheses*{\Omega}$ about the random state, where  $\mu\parentheses*{\omega}$ is the probability that the realized state is $\omega$, and given some set $S$, the simplex $\Delta\parentheses*{S}$ is the set of all distributions over $S$. This makes cheap talk a Bayesian game.

Let $\Sigma$ be an abstract, publicly known, sufficiently rich discrete set of \emph{signals} (e.g., natural language messages up to a certain length).\footnote{The assumption that $\Sigma$ is a discrete set is not essential; we only impose it to facilitate the presentation.} Specifically, we require $\absolute*{\Sigma}\geq n$ (since by Proposition~\ref{pro:signals} below, $n$ signals always suffice for some sender-optimal equilibrium). The game then proceeds as a two-step sequential interaction, as follows:

\begin{itemize}[leftmargin=60pt]
   \item[Step 1] Nature draws a state $\omega \sim \mu$ from the prior distribution. The sender observes $\omega$ and transmits a (possibly random) signal $\sigma \in \Sigma$ to the receiver.
    \item[Step 2] After observing $\sigma$, the receiver deduces a posterior distribution over $\Omega$ and chooses an action $a\in A$. The sender (resp., receiver) gets utility $u_S\parentheses*{\omega,a}$ (resp., $u_R\parentheses*{\omega,a}$).
\end{itemize}

Like the classic cheap talk literature, we analyze this two-step game via the standard equilibrium notion for such Bayesian sequential games, known as perfect Bayesian equilibrium~\cite{fudenberg1991perfect}.

\paragraph{Strategies and beliefs.}
\emph{The sender's (mixed) strategy} in the game is a mapping $\pi:\Omega\to\Delta\parentheses*{\Sigma}$ that maps observed states to distributions over signals; it is called a \emph{signaling policy}. Let $\pi\parentheses*{\sigma\mid \omega}$ denote the probability with which the sender transmits the signal $\sigma$ conditional on observing the state~$\omega$. We use $\pi\brackets*{\omega} \in \Delta\parentheses*{\Sigma}$ to denote the distribution of signals conditional on the state~$\omega$. We write $\sigma \sim \pi$ to denote the \emph{unconditional} distribution of $\sigma$, i.e., this notation describes a two-step process in which we first draw a state $\omega \sim \mu$, and thereafter -- a signal $\sigma\sim\pi\brackets*{\omega}$.

Upon receiving a signal $\sigma$ generated from the sender's signaling policy $\pi$, the receiver infers a \emph{posterior belief} about the underlying state $\omega$ via a standard Bayesian update:

 \begin{equation}\label{eq:posterior-update}     p_{\sigma}\parentheses*{\omega}:=p_{\sigma}^{\pi}\parentheses*{\omega}:=\mathbb{P}\parentheses*{\omega \mid \sigma}= \frac{ \mathbb{P}\parentheses*{\omega,  \sigma} }{\mathbb{P}(\sigma) } = \frac{ \pi\parentheses*{\sigma \mid \omega} \cdot \mu\parentheses*{\omega}}{ \sum_{\omega' \in \Omega} \pi\parentheses*{\sigma \mid \omega'} \cdot \mu\parentheses*{\omega'} }.
\end{equation}

We denote by $p_\sigma:=\parentheses*{p_\sigma\parentheses*{\omega}}_{\omega \in \Omega}$ the $n$-dimensional vector representing the posterior distribution upon receiving signal $\sigma$. The \emph{support} of a distribution $\nu$ over a discrete set $X$ is $\supp\parentheses*{\nu}:=\braces*{x\in X:\;\mathbb{P}_{\nu}\brackets*{x}>0}$. In particular, the support of~$\pi$ is the set of all signals in $\Sigma$ sent with positive probabilities. Notably, for $\sigma\notin \supp\parentheses*{\pi}$, the posterior $p_\sigma$ remains undefined.

\begin{remark}
\label{rem:deviate-in-support}
In the equilibrium notion suitable for cheap talk -- perfect Bayesian equilibrium -- the sender has no incentives to deviate to sending any signal outside of $\supp\parentheses*{\pi}$. More formally, one can define the receiver's posterior belief at off-path signals (i.e., signals outside of $\supp\parentheses*{\pi}$) as inducing the same posterior as some on-path message, which allows to avoid creating beneficial deviations for the~sender.
\end{remark}

\emph{The receiver's (mixed) strategy} is a mapping $s: \Sigma \to \Delta\parentheses*{A}$ that maps signals to distributions over the action set $A$. This strategy is set simultaneously with the sender's strategy and determines how the receiver plays in Step~2. Let $s\parentheses*{a\mid \sigma}$ denote the receiver's probability of taking action $a$ conditional on receiving signal~$\sigma$. We use $s\brackets*{\sigma} \in \Delta \parentheses*{A}$ to denote the distribution of actions conditional on signal~$\sigma$.

\paragraph{Cheap talk equilibrium.}
In cheap talk, no players have commitment power, and thus they set their strategies simultaneously. At equilibrium, a player's strategy must be a \emph{best response} to the opponent's strategy. Consider first the receiver. Given the sender's signal $\sigma$ and the sender's strategy $\pi$ (according to which the receiver performs the Bayesian update as in Equation~\eqref{eq:posterior-update} and deduces posterior $p_\sigma$), the receiver derives his expected utility from each action $a$:

$$
\mathbb{E}_{\omega \sim p_\sigma} \brackets*{u_R\parentheses*{\omega, a}} = \sum_{\omega \in \Omega} p_\sigma\parentheses*{\omega} u_R\parentheses*{\omega, a}.
$$

The receiver's \emph{best response action set} given $\sigma$ and $\pi$ is:
\begin{equation}
\texttt{B}_R\parentheses*{\sigma,\pi} = \bigg\{a \in A : \mathbb{E}_{\omega \sim p_\sigma} \brackets*{u_R\parentheses*{\omega, a}} \geq \mathbb{E}_{\omega \sim p_\sigma} \brackets*{u_R\parentheses*{\omega, a'}}, \forall a' \in A \bigg\}.
\end{equation}

The receiver's distribution over actions $s\brackets*{\sigma}$ is a \emph{best response} to $\sigma$ under $\pi$ if it mixes over only best response actions, i.e.,~$\supp\parentheses*{s\brackets*{\sigma}} \subseteq \texttt{B}_R\parentheses*{\sigma,\pi}$. 
The receiver's strategy $s$ is a \emph{best response to the sender's strategy} $\pi$ if $s\brackets*{\sigma}$ is a best response to $\sigma$ for every $\sigma\in\supp\parentheses*{\pi}$.

The sender's best response requires slightly more careful treatment. Upon observing a state $\omega$, the sender may wish to deviate from sending a signal $\sigma \sim \pi\brackets*{\omega}$ to sending an arbitrary signal $\sigma' \in \Sigma$. However, as noted in Remark~\ref{rem:deviate-in-support}, in cheap talk the sender can only deviate to signals $\sigma'\in\supp\parentheses*{\pi}$. Thus, $\pi$ is a best response if the sender does not benefit from any such deviations under any realized state $\omega$. Formally, given the state $\omega$ and the receiver's strategy $s$, the sender's \emph{best response signal set} is:
\begin{equation}\label{eq:sender-optimal-response}
    \texttt{B}_S\parentheses*{\omega,s} = \bigg\{ \sigma \in \supp\parentheses*{\pi}:   \mathbb{E}_{a 
    \sim s\brackets*{\sigma}} u_S(\omega, a) \geq \mathbb{E}_{a 
    \sim s\brackets*{\sigma'}} u_S\parentheses*{\omega, a} , \forall   \sigma' \in \supp\parentheses*{\pi}  \bigg \}. 
\end{equation}
The sender's signaling policy $\pi$ is a \emph{best response to receiver's strategy} $s$ if $\supp \parentheses*{\pi\brackets*{\omega}} \subseteq \texttt{B}_S\parentheses*{\omega,s}$ for any $\omega$. That is, conditional on any realized state $\omega$, the distribution of the signals $\pi\brackets*{\omega}$ only mixes over the sender's best response signals in the set $\texttt{B}_S\parentheses*{\omega,s}$.

Given the definitions of both players' best responses above, 
we are now ready to define the standard equilibrium concept in the cheap talk model~\cite{crawford1982strategic,FarrellR96,lipnowski2020cheap}:

\begin{definition}[Cheap talk equilibrium] 
Consider a cheap talk game $\langle A,\Omega,\mu,u_S,u_R \rangle$.
A pair of mixed strategies $\parentheses*{\pi, s}$ is a \emph{perfect Bayesian equilibrium} of the game, a.k.a.~a~\emph{cheap talk equilibrium},~if:\footnote{Note that the outcome in the Bayesian persuasion model with sender's commitment power does not require the first constraint on sender's incentives.}
\begin{itemize}
    \item  $\pi$ is a best response to $s$: $\supp \parentheses*{\pi\brackets*{\omega}} \subseteq \texttt{B}_S\parentheses*{\omega,s}$ for every $\omega \in \Omega$; 
    \item  $s$ is a best response to $\pi$: $\supp \parentheses*{s\brackets*{\sigma}} \subseteq \texttt{B}_R\parentheses*{\sigma,\pi}$ for every $\sigma \in \supp\parentheses*{\pi}$. 
\end{itemize}
\end{definition}

While a cheap talk equilibrium always exists, the intractability of computing a Nash equilibrium naturally raises the computational concern of finding one. Interestingly, it turns out that computing one equilibrium in cheap talk is a trivial task. Indeed, revealing no information always yields an equilibrium, which is known as a \emph{babbling equilibrium}.

\begin{observation}[Babbling equilibrium]
\label{obs:no_reveal}
Let $\Sigma = \braces*{\sigma}$ be a singleton set and $\pi\parentheses*{\sigma\mid\omega} = 1\; \forall \omega$ be the no-revelation signaling policy. 
Let $s$ be an action that best responds to the prior $\mu$ (that is played when the only possible signal $\sigma$ is observed). Then the strategy pair $\parentheses*{\pi, s}$ is a cheap talk~equilibrium.
\end{observation}

Therefore, the interesting computational question in cheap talk is optimizing over any \emph{non-trivial} equilibria that may exist in addition to the above no-information ones.

\paragraph{Example.} Consider a cheap talk instance with a binary state space $\Omega=\braces*{0,1}$ (i.e., $n=2$), a uniform prior distribution, $m=4$ actions and the utilities specified by Table~\ref{tab:bin}.\footnote{Examples with a similar spirit are considered, e.g., in~\cite{forges1994non,forges2020games}.}

\begin{table}[h]
\begin{center}
\begin{tabular}{|c|c|c|c|c|}
\hline
$u_S$      & $a_1$ & $a_2$ & $a_3$ & $a_4$ \\ \hline
$\omega=0$ & $-1$     & $2$     & $-2$    & $3$     \\ \hline
$\omega=1$ & $3$    & $-2$    & $2$     & $-1$     \\ \hline
\end{tabular}
\hspace{5mm}
\begin{tabular}{|c|c|c|c|c|}
\hline
$u_R$      & $a_1$ & $a_2$ & $a_3$ & $a_4$ \\ \hline
$\omega=0$ & $3$     & $2$     & $-2$    & $-5$     \\ \hline
$\omega=1$ & $-5$     & $-2$    & $2$     & $3$     \\ \hline
\end{tabular}
\end{center}
\caption{The utilities $u_S$ and $u_R$.}\label{tab:bin}
\end{table}

Even though we are interested in the cheap talk model, it will be useful to first illustrate the Bayesian persuasion analysis in this example. Due to binary states, any posterior distribution $p_\sigma$ is fully determined by the probability $p_\sigma\parentheses*{\omega_1}$, since  $p_\sigma\parentheses*{\omega_0} +  p_\sigma\parentheses*{\omega_1} = 1$. Thus, in the following discussion, we shall identify the posterior with $p_\sigma\parentheses*{\omega_1}$ to allow a more convenient geometric interpretation.  Figure~\ref{fig:receiver-br}(a) shows that the action $a_1$ is a receiver's best response when the posterior is in the segment $\brackets*{0,\frac{1}{4}}$, the action $a_2$ is a best response when the posterior is in $ \brackets*{\frac{1}{4},\frac{1}{2}}$, $a_3$ is a best response in $\brackets*{\frac{1}{2},\frac{3}{4}}$ and $a_4$ is a best response in $\brackets*{\frac{3}{4},1}$. The \emph{sender's indirect utility} (i.e., sender's utility as a function of receiver's posterior) is as demonstrated in Figure~\ref{fig:receiver-br}(b). 

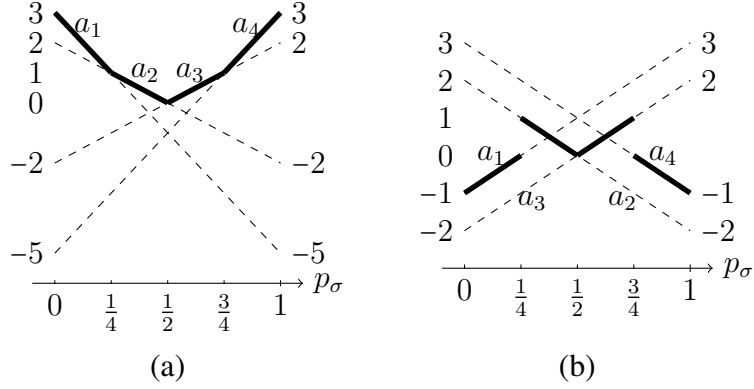
\begin{figure}[h]
    \centering
    \begin{tikzpicture}[xscale=3,yscale=0.4]
       \draw[->] (-0.1,-6) -- (1.1,-6);
       \node[right] at (1.1,-6) {$p_{\sigma}$};

       \draw[dashed] (0,3) -- (1,-5);
       \draw[line width=2] (0,3)--(0.25,1);

       \draw[dashed] (0,2) -- (1,-2);
       \draw[line width=2] (0.25,1) -- (0.5,0);

       \draw[dashed] (0,-2) -- (1,2);
       \draw[line width=2] (0.75,1) -- (0.5,0);

       \draw[dashed] (0,-5) -- (1,3);
       \draw[line width=2] (0.75,1) -- (1,3);      

       \draw (0,-5.9)--(0,-6.1);
       \node[below, fill=white] at (0,-6) {$0$};

       \draw (0.25,-5.9)--(0.25,-6.1);
       \node[below, fill=white] at (0.25,-6) {$\frac{1}{4}$};

       \draw (0.5,-5.9)--(0.5,-6.1);
       \node[below, fill=white] at (0.5,-6) {$\frac{1}{2}$};

       \draw (0.75,-5.9)--(0.75,-6.1);
       \node[below, fill=white] at (0.75,-6) {$\frac{3}{4}$};

        \draw (1,-5.9)--(1,-6.1);
       \node[below] at (1,-6) {$1$};

       \node[left] at (0,3) {$3$};
       \node[left] at (0,2) {$2$};
       \node[left] at (0,1) {$1$};
       \node[left] at (0,0) {$0$};
       \node[left] at (0,-2) {$-2$};
       \node[left] at (0,-5) {$-5$};

       \node[right] at (1,3) {$3$};
       \node[right] at (1,2) {$2$};
       \node[right] at (1,-2) {$-2$};
       \node[right] at (1,-5) {$-5$};

       \node[below] at (0.5,-8) {(a)};

       \node at (0.15,2.5) {$a_1$};

       \node at (0.4,1) {$a_2$};

       \node at (0.6,1) {$a_3$};

       \node at (0.85,2.5) {$a_4$};
              
    \end{tikzpicture}
    \hspace{5mm}
    \begin{tikzpicture}[xscale=3,yscale=0.5]
       \draw[->] (-0.1,-3) -- (1.1,-3);
       \node[right] at (1.1,-3) {$p_{\sigma}$};

       \draw[dashed] (0,-1) -- (1,3);
       \draw[line width=2] (0,-1)--(0.25,0);

       \draw[dashed] (0,2) -- (1,-2);
       \draw[line width=2] (0.25,1) -- (0.5,0);

       \draw[dashed] (0,-2) -- (1,2);
       \draw[line width=2] (0.75,1) -- (0.5,0);

       \draw[dashed] (0,3) -- (1,-1);
       \draw[line width=2] (0.75,0) -- (1,-1);
       
       \draw (0,-2.9)--(0,-3.1);
       \node[below, fill=white] at (0,-3) {$0$};

       \draw (0.25,-2.9)--(0.25,-3.1);
       \node[below, fill=white] at (0.25,-3) {$\frac{1}{4}$};

       \draw (0.5,-2.9)--(0.5,-3.1);
       \node[below, fill=white] at (0.5,-3) {$\frac{1}{2}$};

       \draw (0.75,-2.9)--(0.75,-3.1);
       \node[below, fill=white] at (0.75,-3) {$\frac{3}{4}$};

       \draw (1,-2.9)--(1,-3.1);
       \node[below] at (1,-3) {$1$};

       \node[left] at (0,3) {$3$};
       \node[left] at (0,2) {$2$};
        \node[left] at (0,1) {$1$};
        \node[left] at (0,0) {$0$};
        
       \node[left] at (0,-2) {$-2$};
       \node[left] at (0,-1) {$-1$};
       
       \node[right] at (1,3) {$3$};
       \node[right] at (1,2) {$2$};
       \node[right] at (1,-1) {$-1$};
       \node[right] at (1,-2) {$-2$};
    
       \node[below] at (0.5,-5) {(b)};

       \node at (0.12,0) {$a_1$};

       \node at (0.88,0) {$a_4$};       
       
       \node at (0.7,-1.2) {$a_2$};

       \node at (0.3,-1.2) {$a_3$};
       
    \end{tikzpicture}
    \caption{Figure~(a) demonstrates the receiver's utility as a function of the posterior probability $p_{\sigma}\parentheses*{\omega_1}$ (the horizontal axis). Dashed lines capture the utilities of all actions; solid lines capture the utility of the receiver's best action. Figure~(b) depicts the sender's indirect utility (i.e., her utility as a function of the receiver's posterior $p_{\sigma}\parentheses*{\omega_1}$). Again, dashed lines capture the utilities of all actions, while solid lines capture the sender's utility under a receiver's best response action.}
    \label{fig:receiver-br}
\end{figure}

The Bayesian persuasion solution is the concavification of the sender's indirect utility (see~\cite{kamenica2011bayesian}) evaluated at the prior. In this case, it equals $1$. The 
equilibrium (with sender's commitment) that arises is as follows: The sender splits the prior $\mu=\frac{1}{2}$ into the two posteriors $p_\sigma=\frac{1}{4}$ and $p_\sigma=\frac{3}{4}$ with equal probability $\frac{1}{2}$. As is standard in Bayesian persuasion, we assume that the receiver breaks ties in the sender's favor. Consequently, the receiver plays action $a_2$ when his posterior is $\frac{1}{4}$, and he plays action $a_3$ when his posterior is $\frac{3}{4}$.

Note, however, that the above profile does \emph{not} constitute an equilibrium in the cheap talk model (without commitment). At state $\omega=0$, the sender will prefer sending the signal inducing the posterior $\frac{1}{4}$ with probability $1$ (instead of mixing between the two signals). Similar inconsistency arises at state $\omega=1$ in which the receiver prefers the posterior $\frac{3}{4}$.

Nevertheless, the above equilibrium with sender's commitment can be carefully modified to form a cheap talk equilibrium. Notice that for the posterior $\frac{1}{4}$, the receiver is indifferent between the actions $a_1$ and $a_2$, and for the posterior $\frac{3}{4}$ the receiver is indifferent between the actions $a_3$ and $a_4$. Can the receiver play mixed strategies under these posteriors causing the sender to be indifferent between the two signals in both states? In this example, the answer is affirmative. If at the posterior $\frac{1}{4}$ the receiver plays both actions $a_1$ and $a_2$ with equal probability $\frac{1}{2}$, and at the posterior $\frac{3}{4}$ he plays both $a_3$ and $a_4$ with equal probability $\frac{1}{2}$, then the sender is indifferent between inducing both posteriors in each state. Indeed, at state $\omega=0$, the posterior $\frac{1}{4}$ will yield her an expected utility of $\frac{1}{2}\cdot 2 + \frac{1}{2} \cdot \parentheses*{-1}=\frac{1}{2}$, and the posterior $\frac{3}{4}$ will yield her an expected utility of $\frac{1}{2}\cdot \parentheses*{-2}+\frac{1}{2}\cdot 3 = \frac{1}{2}$. Similarly, one may verify it for the state $\omega=1$. It turns out that the aforementioned equilibrium is sender-optimal, with an expected sender's utility of $1/2$.\footnote{To verify its optimality, one can apply an exhaustive search on all possible supports of the equilibrium, which by Proposition~\ref{pro:signals} proof can be assumed to be binary.}

\subsection{A Bound on the Number of Signals}
\label{sub:signals}
We now show that, as far as the sender's equilibrium utility is concerned,  any cheap talk equilibrium needs not to use more than $n$ signals. The proof uses Carathéodory's Theorem. A similar result holds in the Bayesian persuasion setting as well, but in cheap talk, it requires slightly more delicate~arguments.

\begin{proposition}
\label{pro:signals}
For any expected sender's utility value $v$ that can arise at some cheap talk equilibrium, there exists a cheap talk equilibrium $\parentheses*{\pi,s}$ with $\absolute*{\supp\parentheses*{\pi}}\leq n$ leading to expected sender's utility of $v$.
\end{proposition}

\begin{proof}
Given an equilibrium $\parentheses*{\pi,s}$ with an arbitrary support size $\absolute*{\supp\parentheses*{\pi}}> n$ (or an infinite support $\absolute*{\supp\parentheses*{\pi}}= \infty$), we construct another equilibrium $\parentheses*{\pi',s}$ with $\absolute*{\supp\parentheses*{\pi'}}\leq n$ and with the same expected sender's utility. The receiver's strategy $s$ will remain unchanged. 

Note that the posterior beliefs $p_\sigma$, where $\sigma \sim \pi$, must form a split of the prior (this follows from the law of total probability). That is, we have $\mathbb{E}_{\sigma \sim \pi}\brackets*{p_\sigma}=\mu$. By Carathéodory's theorem, there exists another split $\nu$ of the prior with support size at most $n$ that uses only posteriors from $\braces*{p_\sigma}_{\sigma \in \supp \parentheses*{\pi}}$. We have $\absolute*{\supp\parentheses*{\nu}}\leq n$ and $\mathbb{E}_{p \sim \nu}\brackets*{p}=\mu$. The well-known splitting Lemma of Aumann and Maschler~\cite{aumann1995repeated} implies that there exists a signaling policy $\pi'$ whose induced distribution over posteriors is $\nu$, and the signaling policy $\pi'$ uses at most $n$ signals.

To see that $\parentheses*{\pi',s}$ is a cheap talk equilibrium, note that:

    \begin{itemize}
        \item Conditional on any state $\omega$, sender's best response signals are $\braces*{p_\sigma:\sigma \in \supp\parentheses*{\pi}, p_{\sigma}\parentheses*{\omega}>0}$. Under the strategy $\pi'$, the sender transmits a subset of these signals, as needed.
        \item In every posterior $p_\sigma$ for $\sigma \in \supp\parentheses*{\pi'}$, the receiver best responds. This is because both the posterior and the receiver's strategy remain unchanged compared to the equilibrium $\parentheses*{\pi,s}$. 
    \end{itemize}

    To see that the sender's expected utility remains unchanged, we observe that the sender is \emph{indifferent} between all the posteriors $\braces*{p_\sigma:\sigma \in \supp\parentheses*{\pi}, p_{\sigma}\parentheses*{\omega}>0}$ conditional on the state $\omega$. She uses a subset of the signals inducing these posteriors at the new equilibrium $\parentheses*{\pi',s}$; as all the posteriors induced by these signals, as well as the receiver's responses, are unchanged -- the sender's expected utility remains~unchanged.   
\end{proof}

By Proposition~\ref{pro:signals}, there always exists a sender-optimal equilibrium in which the sender uses at most $n$ signals. We immediately obtain a $\poly\parentheses*{m,n}$-sized witness for the sender's ability to ensure a certain amount of expected utility at equilibrium. Therefore, we have the following~corollary.

\begin{corollary}[NP Membership of Cheap Talk]
    \label{cor:NP}
    Given a cheap talk instance $\langle A,\Omega,\mu,u_S,u_R \rangle$ and a value $v$, the problem of deciding whether there exists an equilibrium yielding an expected sender's utility of at least $v$ belongs to NP.
\end{corollary}

\section{The Hardness of Cheap Talk}
\label{sec:hard}
Following the approximation algorithm convention, when an equilibrium yields an expected sender's utility within an additive constant $c>0$ of the best possible, we say it is \emph{$c$-sender-optimal} (in the additive sense). Our main hardness result is as follows.

\begin{theorem}\label{thm:additive}
There exists an absolute constant $c > 0$ such that it is NP-hard to compute a $c$-sender-optimal equilibrium in the additive sense for normal form cheap talk games with normalized utilities in $\brackets*{0,1}$.\footnote{In our proof, the constant is  $c = \frac{1}{113792}$. We did not try to optimize this constant, though this gap can be increased by using the inapproximability of other $d$-regular $k$SAT instances (larger $d$ and smaller $k$ are preferred). However, even though this may lead to a better (larger) constant, it is a highly non-trivial open question to identify the tight~constant.} 
\end{theorem}

Theorem~\ref{thm:additive} immediately implies the same hardness result for multiplicative approximation. If an equilibrium yields expected sender's utility within a multiplicative factor $0<c<1$ of the best possible, we say it is \emph{$c$-sender-optimal in the multiplicative sense.} Since the utilities in Theorem~\ref{thm:additive} are normalized in $\brackets*{0,1}$, we immediately deduce the following corollary.

\begin{corollary}
    \label{rem:multi}
    There exists an absolute constant $0 < c <1$ such that it is NP-hard to compute a $c$-sender-optimal equilibrium in the multiplicative sense for normal form cheap talk games with non-negative~utilities.
\end{corollary}

A corollary of the proof of Theorem~\ref{thm:additive} addresses the question: Can the sender do better than revealing no information (which results in a babbling equilibrium)?

\begin{proposition}\label{pro:babbling}
    It is NP-complete to decide whether there exists a cheap talk equilibrium yielding higher expected sender's utility than a babbling equilibrium.
\end{proposition}

Proposition~\ref{pro:babbling}, though straightforward, uses some concrete characteristics of the reduction rather than using Theorem~\ref{thm:additive} as a black box. Thus, we relegate the proof to Appendix~\ref{append:prop:babbling}.

The proof of Theorem~\ref{thm:additive} is non-standard compared to existing hardness results on the complexity of equilibria -- see Subsection~\ref{sub:related}. The remainder of this section is devoted to the formal proof. We start with an overview of our proof structure and then dive into the details of each proof~step.

\begin{proof}[Theorem~\ref{thm:additive} proof overview]
At a high level, our proof has two major steps: (1)~identifying a suitable NP-hard problem instance to reduce from; (2)~establishing the reduction. Regarding step~(1), it appears challenging to establish an approximation-preserving reduction from classic NP-hard problems to cheap talk games. We thus introduced a new variant of the \texttt{Max-3SAT} problem, termed \texttt{Max-Var-3SAT}, defined~below. 

Recall that the classic \texttt{Max-3SAT} problem aims to find a Boolean variable assignment to maximize the number of satisfied clauses in a given 3CNF Boolean formula.\footnote{A Boolean formula is 3CNF (i.e., $3$-conjunctive normal form) if it is given by $C_1\wedge\dots \wedge C_m$, where each \emph{clause} $C_j$ is a disjunction ($\vee$) of exactly $3$ \emph{literals} (a literal is a variable $x_i$ or its negation $\neg x_i$), and no variable appears twice in a clause.} The problem instance we shall use is a variable-maximizing variant of this problem, hence the name \texttt{Max-Var-3SAT}, which maximizes the number of assigned Boolean variables without creating any contradictory~clauses.

Formally, let $\phi$ be a 3CNF formula with $n$ variables $ x_1,\ldots,x_n $  and $m$ clauses. A \emph{partial assignment} $x_S$ to $k$ variables is specified by a subset of variables $S\subseteq \brackets*{n} = \braces*{ 1, 2, \cdots, n}$ of size $k$, as well as an assignment $x_i\in \braces*{\texttt{True},\texttt{False} }$ for every $ i \in S$.  We say that a clause is \emph{contradictory} (to satisfiability) under partial assignment $x_S$ if $x_S$ has assigned values to all the 3 variables associated with this clause and their values evaluate the clause to \texttt{False}. The partial assignment $x_S$ is said to be  \emph{non-contradictory}  if it does not create any contradictory clauses. We are now ready to define the  \texttt{Max-Var-3SAT} problem. 
 \begin{definition} The \texttt{Max-Var-3SAT} problem  takes any 3CNF formula $\phi$ as input and outputs the largest integer $k$ such that there exists a non-contradictory partial assignment $x_S$ of size $k$ (i.e., $|S| = k$). 
 \end{definition}

 \texttt{Max-Var-3SAT} is clearly NP-hard since its decision variant of determining whether there exists a non-contradictory partial assignment of size $n$ or not is precisely the   \texttt{3SAT} problem. However, for our reduction, we shall need a stronger inapproximability result for a restricted version of \texttt{Max-Var-3SAT}. We say that a \texttt{Max-Var-3SAT} instance is \emph{$d$-regular} if, in the input 3CNF formula, every variable appears in exactly $d$ clauses. In Appendix \ref{append:prop:sat-hard},  we present a proof for the inapproximability of the following problem for $4$-regular \texttt{Max-Var-3SAT}.

\begin{proposition}\label{pro:regular}
Suppose that a $4$-regular \texttt{Max-Var-3SAT} instance is promised to satisfy one of the following two conditions:
\begin{enumerate}
    \item \texttt{Max-Var-3SAT}$\parentheses*{\phi}\geq \frac{3047.6}{3048} n $.
    \item \texttt{Max-Var-3SAT}$\parentheses*{\phi}< \frac{3047.1}{3048} n$.
\end{enumerate}
It is NP-complete to decide whether the instance satisfies the first condition or not.
\end{proposition}

The second, and also the most involved step, is to establish an approximation-preserving reduction from the above problem to cheap talk equilibrium computation. Our key conclusion is summarized in the following proposition.

\begin{proposition}\label{pro:reduction}
Suppose $d\leq 6$ and let $\phi$  be any $d$-regular 3CNF formula  with $n$ variables and $m$ clauses. Then there exists a $\poly\parentheses*{m,n}$-time algorithm that takes $\phi$ as input and outputs a cheap talk instance $C$ with $7m$ states, $O\parentheses*{mn}$ receiver's actions and state-dependent sender's utilities in $\brackets*{-7,1}$, such that the best expected utility that the sender can achieve in $C$ is $\frac{kd}{7m}$, where  $k = \texttt{Max-Var-3SAT}\parentheses*{\phi}$ is the solution to the $d$-regular \texttt{Max-Var-3SAT} problem.   
\end{proposition}

The proof of Proposition~\ref{pro:reduction} appears below.  Before that, let us show how the two propositions~\ref{pro:regular} and~\ref{pro:reduction} indeed complete the proof of Theorem~\ref{thm:additive}. Recall that for approximation, one should normalize the sender's utilities to belong to $\brackets*{0,1}$. For the instance $C$ obtained from the reduction in Proposition~\ref{pro:reduction}, this normalization can be done by first adding the constant $7$ to each entry of the sender's utility and then dividing them by $8$. Since such transformation would not change the set of equilibria, the corresponding sender-optimal cheap talk equilibrium now has expected sender's utility of $\parentheses*{\frac{kd}{7m} + 7}/8 = \frac{kd}{56m }+ \frac{7}{8}$, where $k =\texttt{Max-Var-3SAT}\parentheses*{\phi}$.  By Proposition~\ref{pro:regular}, it is NP-hard to decide between the two cases $$k\geq \frac{3047.6}{3048}n := q_1 n \text{\ \  and \ \ } k < \frac{3047.1}{3048} n := q_2 n.$$ Note that we have $n = \frac{3}{4}m$ for $4$-regular \texttt{Max-Var-3SAT}  instances ($d=4$), which implies that the expected sender's utility difference between the above two situations is at least $$ \frac{q_1 n \times d}{56m } -  \frac{q_2n \times d}{56m } = \frac{\parentheses*{q_1-q_2}4 \times 3/4}{56},$$ which equals $c = 1/113792$ for $d=4$. This proves that it is NP-hard to compute a $c$-sender-optimal equilibrium, in the additive sense, for the constructed instance.
\end{proof}

\subsection{Proof of Proposition~\ref{pro:reduction}: An Approximation-Preserving Reduction}
\label{sub:reduction}
In this subsection, we present the main part of our proof, which is an approximation-preserving reduction from $d$-regular~\texttt{Max-Var-3SAT} to the computation of a sender-optimal cheap talk equilibrium. Due to the intricacy of the arguments, we present the proof in three steps.

\subsubsection*{Step 1: Constructing the Cheap Talk Instance}
We start by constructing the cheap talk instance $\langle A, \Omega, \mu, u_S, u_R \rangle $ from any given $d$-regular 3CNF formula $\phi$ with $n$ variables and $m$ clauses. A graphical illustration of our construction can be found in Figure~\ref{fig:enter-label}. 

\vspace{2mm}
\noindent   \textbf{States $\Omega$ and prior $\mu$.}  For each variable $x_i$ that appears (possibly with negation) in the $j$th clause, we create two \emph{variable states} denoted as $x_{i,j}$ and $\neg x_{i,j}$. Additionally, we create a single \emph{clause state} $c_j$ for every clause $j\in \brackets*{m}$. 
In total, there are $7m$ states -- for every clause $j\in \brackets*{m}$ there are $6$ states for its three variables, as well as their negations, and one additional clause state for itself. Let $\Omega$ denote the set of all these $7m$ states. The prior $\mu$ over $\Omega$ is simply the uniform distribution. The constructions of the receiver's and sender's utilities are more delicate and hinge on a collection of special posterior distributions -- i.e., points in the simplex $\Delta\parentheses*{\Omega}$ -- that will play a special role in our reduction. We consider the simplex $\Delta\parentheses*{\Omega}$ to be embedded into $\mathbb{R}^{7m}$, with the $i$-th coordinate representing the probability for the state being $\omega_i$.

 \begin{figure}
     \centering
     \begin{tikzpicture}[scale=0.45]
\draw[rounded corners] (0, 0) rectangle (8, -1) {};
\node at (3.5,-0.5) {$x_{1,1}$};
\draw[rounded corners] (9, 0) rectangle (17, -1) {};
\node at (12.5,-0.5) {$x_{1,2}$};

\draw[blue, rounded corners]  (-3,0.2) rectangle (17.2,-1.2) {};
\node at (-1.7,-0.5) {$P_v(1)$};

\draw[rounded corners] (0, -2) rectangle (8, -3) {};
\node at (3.5,-2.5) {$\neg x_{1,1}$};
\draw[rounded corners] (9, -2) rectangle (17, -3) {};
\node at (12.5,-2.5) {$\neg x_{1,2}$};

\draw[blue, rounded corners]  (-3,-1.8) rectangle (17.2,-3.2) {};
\node at (-1.7,-2.5) {$P_v(\neg 1)$};

\draw[rounded corners] (0, -4) rectangle (8, -5) {};
\node at (3.5,-4.5) {$x_{2,1}$};
\draw[rounded corners] (9, -4) rectangle (17, -5) {};
\node at (12.5,-4.5) {$x_{2,2}$};

\draw[blue, rounded corners]  (-3,-3.8) rectangle (17.2,-5.2) {};
\node at (-1.7,-4.5) {$P_v(2)$};

\draw[rounded corners] (0, -6) rectangle (8, -7) {};
\node at (3.5,-6.5) {$\neg x_{2,1}$};
\draw[rounded corners] (9, -6) rectangle (17, -7) {};
\node at (12.5,-6.5) {$\neg x_{2,2}$};

\draw[blue, rounded corners]  (-3,-5.8) rectangle (17.2,-7.2) {};
\node at (-1.7,-6.5) {$P_v(\neg 2)$};

\draw[rounded corners] (0, -8) rectangle (8, -9) {};
\node at (3.5,-8.5) {$x_{3,1}$};

\draw[blue, rounded corners]  (-3,-7.8) rectangle (17.2,-9.2) {};
\node at (-1.7,-8.5) {$P_v(3)$};

\draw[rounded corners] (0, -10) rectangle (8, -11) {};
\node at (3.5,-10.5) {$\neg x_{3,1}$};

\draw[blue, rounded corners]  (-3,-9.8) rectangle (17.2,-11.2) {};
\node at (-1.7,-10.5) {$P_v(\neg 3)$};

\draw[rounded corners] (9, -12) rectangle (17, -13) {};
\node at (12.5,-12.5) {$x_{4,2}$};

\draw[blue, rounded corners]  (-3,-11.8) rectangle (17.2,-13.2) {};
\node at (-1.7,-12.5) {$P_v(4)$};

\draw[rounded corners] (9, -14) rectangle (17, -15) {};
\node at (12.5,-14.5) {$\neg x_{4,2}$};

\draw[blue, rounded corners]  (-3,-13.8) rectangle (17.2,-15.2) {};
\node at (-1.7,-14.5) {$P_v(\neg 4)$};

\draw[rounded corners] (0, -16) rectangle (8, -17) {};
\node at (3.5,-16.5) {$c_1$};
\draw[rounded corners] (9, -16) rectangle (17, -17) {};
\node at (12.5,-16.5) {$c_2$};

\node at (2,-19.5) {The pools $P_c(1,1),\ldots,P_c(1,7)$};
\draw[dashed, ->] (3,-19)--(0.6,-16.8);
\draw[dashed, ->] (3,-19)--(1.5,-16.8);
\draw[dashed, ->] (3,-19)--(2.4,-16.8);

\draw[dashed, ->] (3,-19)--(4.6,-16.8);
\draw[dashed, ->] (3,-19)--(5.6,-16.8);
\draw[dashed, ->] (3,-19)--(6.5,-16.8);
\draw[dashed, ->] (3,-19)--(7.4,-16.8);

\node at (15,-19.5) {The pools $P_c(2,1),\ldots,P_c(2,7)$};
\draw[dashed, ->] (14,-19)--(9.7,-16.8);
\draw[dashed, ->] (14,-19)--(10.7,-16.8);
\draw[dashed, ->] (14,-19)--(11.6,-16.8);

\draw[dashed, ->] (14,-19)--(13.8,-16.8);
\draw[dashed, ->] (14,-19)--(14.6,-16.8);
\draw[dashed, ->] (14,-19)--(15.5,-16.8);
\draw[dashed, ->] (14,-19)--(16.4,-16.8);

\filldraw (6.6,-2.5) circle(0.2);
\filldraw (7.5,-2.5) circle(0.2);

\filldraw (0.5,-0.5) circle(0.2);
\filldraw (0.5,-4.5) circle(0.2);
\filldraw (0.5,-8.5) circle(0.2);
\filldraw (0.5,-16.5) circle(0.2);
\draw (0.5,-0.5)--(0.5,-16.5);

\filldraw (1.4,-0.5) circle(0.2);
\filldraw (1.4,-4.5) circle(0.2);
\filldraw (1.4,-10.5) circle(0.2);
\filldraw (1.4,-16.5) circle(0.2);
\draw (1.4,-0.5)--(1.4,-16.5);

\filldraw (2.3,-0.5) circle(0.2);
\filldraw (2.3,-6.5) circle(0.2);
\filldraw (2.3,-8.5) circle(0.2);
\filldraw (2.3,-16.5) circle(0.2);
\draw (2.3,-0.5)--(2.3,-16.5);

\filldraw (4.8,-0.5) circle(0.2);
\filldraw (4.8,-6.5) circle(0.2);
\filldraw (4.8,-10.5) circle(0.2);
\filldraw (4.8,-16.5) circle(0.2);
\draw (4.8,-0.5)--(4.8,-16.5);

\filldraw (5.7,-2.5) circle(0.2);
\filldraw (5.7,-4.5) circle(0.2);
\filldraw (5.7,-8.5) circle(0.2);
\filldraw (5.7,-16.5) circle(0.2);
\draw (5.7,-2.5)--(5.7,-16.5);

\filldraw (6.6,-2.5) circle(0.2);
\filldraw (6.6,-4.5) circle(0.2);
\filldraw (6.6,-10.5) circle(0.2);
\filldraw (6.6,-16.5) circle(0.2);
\draw (6.6,-2.5)--(6.6,-16.5);

\filldraw (7.5,-2.5) circle(0.2);
\filldraw (7.5,-6.5) circle(0.2);
\filldraw (7.5,-10.5) circle(0.2);
\filldraw (7.5,-16.5) circle(0.2);
\draw (7.5,-2.5)--(7.5,-16.5);

\filldraw (9.5,-0.5) circle(0.2);
\filldraw (9.5,-4.5) circle(0.2);
\filldraw (9.5,-12.5) circle(0.2);
\filldraw (9.5,-16.5) circle(0.2);
\draw (9.5,-0.5)--(9.5,-16.5);

\filldraw (10.4,-0.5) circle(0.2);
\filldraw (10.4,-4.5) circle(0.2);
\filldraw (10.4,-14.5) circle(0.2);
\filldraw (10.4,-16.5) circle(0.2);
\draw (10.4,-0.5)--(10.4,-16.5);

\filldraw (11.3,-0.5) circle(0.2);
\filldraw (11.3,-6.5) circle(0.2);
\filldraw (11.3,-12.5) circle(0.2);
\filldraw (11.3,-16.5) circle(0.2);
\draw (11.3,-0.5)--(11.3,-16.5);

\filldraw (13.8,-2.5) circle(0.2);
\filldraw (13.8,-4.5) circle(0.2);
\filldraw (13.8,-12.5) circle(0.2);
\filldraw (13.8,-16.5) circle(0.2);
\draw (13.8,-2.5)--(13.8,-16.5);

\filldraw (14.7,-2.5) circle(0.2);
\filldraw (14.7,-4.5) circle(0.2);
\filldraw (14.7,-14.5) circle(0.2);
\filldraw (14.7,-16.5) circle(0.2);
\draw (14.7,-2.5)--(14.7,-16.5);

\filldraw (15.6,-2.5) circle(0.2);
\filldraw (15.6,-6.5) circle(0.2);
\filldraw (15.6,-12.5) circle(0.2);
\filldraw (15.6,-16.5) circle(0.2);
\draw (15.6,-2.5)--(15.6,-16.5);

\filldraw (16.5,-2.5) circle(0.2);
\filldraw (16.5,-6.5) circle(0.2);
\filldraw (16.5,-14.5) circle(0.2);
\filldraw (16.5,-16.5) circle(0.2);
\draw (16.5,-2.5)--(16.5,-16.5);

     \end{tikzpicture}
     \caption{The states, variable pools and clause pools of the $3$CNF formula with the variables $x_1,x_2,x_3,x_4$ and the clauses $x_1 \lor x_2 \lor \neg x_3$ and $\neg x_1 \lor x_2 \lor x_4$. The states are depicted as black rectangles; the variable pools are blue rectangles; the clause pools are lines with dots, with a dot indicating that the state belongs to the pool.} 
     \label{fig:enter-label}
 \end{figure}
 
\vspace{2mm}
\noindent   \textbf{Special posteriors -- pools.} For every variable $x_i$,  the uniform distribution over the states $\braces*{x_{i,j}}_{j}$ will be called \emph{the $x_i$ variable pool} and denoted by $P_v\parentheses*{i}$. Note that  $P_v\parentheses*{i}$ is a uniform distribution over $d$ states. Thus, we shall also refer to $P_v\parentheses*{i}$ as a point in the simplex $\Delta\parentheses*{\Omega}$. The name \emph{pool} captures the idea that if the sender decides to pool together the states $\braces*{x_{i,j}}_j$ (namely, to send the same deterministic signal in all these states only) then receiver's posterior belief will be exactly the distribution $\texttt{U}\parentheses*{\braces*{x_{i,j}}_{j}}$, where $\texttt{U}\parentheses*{S}$ denotes the uniform distribution over a subset $S$ of states.\footnote{The same posterior can arise in randomized signaling policies by pooling together fractions of the corresponding~states.} Similarly, we define \emph{the $\neg x_i$ variable pool} $P_v\parentheses*{\neg i}$ to be the uniform distribution over the states $\braces*{\neg x_{i,j}}_{j}$. Let $P_v := \braces*{P_v\parentheses*{i},P_v\parentheses*{\neg i}}_{i\in \brackets*{n}}$ denote the set of all variable pools.

We now describe \emph{clause pools}. For every clause $j\in \brackets*{m}$, there are 7 \texttt{True}/\texttt{False}  assignments for its variables $x_{i_1},x_{i_2},x_{i_3}$ that will satisfy the clause. For each one of these 7 assignments, we create a \emph{clause pool} that is the uniform distribution over four states -- the state $c_j$ and the three corresponding states $x_{i_l,j}$ or $\neg x_{i_l,j}$ based on whether $x_{i_l}=$\texttt{True} or $x_{i,l}=$\texttt{False}, respectively. For example, if $x_{i_1}=$\texttt{True}, $x_{i_2}=$\texttt{False}, $x_{i_3}=$\texttt{True} is a satisfying assignment for the $j$th clause,  then it induces the clause pool $\texttt{U}\parentheses*{\braces*{c_j,x_{i_1,j},\neg x_{i_2,j}, x_{i_3,j}}}$. These pools are denoted by $P_c\parentheses*{j,1},\ldots,P_c\parentheses*{j,7}\in \Delta\parentheses*{\Omega}$ for every $j\in \brackets*{m}$. We denote by $P_c := \braces*{P_c\parentheses*{j,1},\ldots,P_c\parentheses*{j,7}}_{j\in \brackets*{m}}$ the set of all clause pools.

We further introduce \emph{singleton variable pools} that are the Dirac distributions over some $x_{i,j}$ and the Dirac distributions over some $\neg x_{i,j}$.\footnote{A \emph{Dirac distribution} assigns probability $1$ to some element.} These pools are denoted by $P_s\parentheses*{i,j}\in \Delta\parentheses*{\Omega}$ and $P_s\parentheses*{\neg i,j}\in \Delta\parentheses*{\Omega}$, respectively. Let $P_s := \braces*{P_s\parentheses*{i,j},P_s\parentheses*{\neg i,j}}_{i,j}$ (where $i,j$ run over all $j\in\brackets*{m}$ and $i\in\brackets*{n}$ s.t.~the $j$th clause contains the variable $x_i$ or its negation) be the set of all singleton variable pools. Importantly, there are no singleton pools over the clause states $c_j$. The collection of all types of pools is denoted by $P := P_v \cup P_c \cup P_s$.

\vspace{2mm}
\noindent   \textbf{Receiver's actions $A$ and utility $u_R$.} At a high level, we would like to design Receiver's actions and utilities in a way that induces the following situation: Whenever the receiver's posterior $p$ is not one of the pools -- i.e., $p\notin P$ -- then any receiver's best response action will be a ``bad action’’ for the sender; yet for every posterior distribution $p\in P$, there will be a special action $a_p$ that will serve as a best response only at the point $p$; the actions $a_p$ might not be bad for the sender, depending on the state. We shall show that it is indeed possible to achieve these requirements using polynomially many actions for the receiver. Since this result about a decision maker under uncertainty (the receiver) may be of independent interest, we state it as a proposition. 

\begin{proposition}\label{pro:rec-ut}
    Given a finite state space $\Omega$ and  any finite collection of points $P\subseteq \Delta\parentheses*{\Omega}$, there exists a polynomial-time algorithm that outputs a collection of $\parentheses*{\absolute*{\Omega}+1}\absolute*{P}$ actions for the receiver $A=\{a_p\}_{p\in P}\cup \{a_{p,\omega}\}_{p\in P, \omega \in \Omega}$ and  receiver's utilities $u_R:\Omega \times A \to \mathbb{R}$ with the following properties:
    \begin{enumerate}
        \item Under any posterior $y=p\in P$, the receiver's action $a_p$ is a best response and any other action $a_{p'}$  for $p'\neq p$ is  not a best response.\footnote{Notably, some $a_{p',\omega}$ may also be best responses under posterior belief $p$. This property only ensures that $a_{p'}$  will not be a best response under $p\neq p'$, and this will suffice for our later argument.}
        \item Under every posterior $y\notin P$, none of the actions $a_p$ for $p\in P$ is a best response.
    \end{enumerate}
\end{proposition}

In the proposition above, actions $a_{p,\omega}$ will play the role of ``bad actions’’ for the sender, which will be used by the receiver when $y \notin P$. To construct the receiver's actions and utilities in our reduction, we apply Proposition~\ref{pro:rec-ut} with $P$ being the set of all the constructed \emph{pools} above. It will incentivize the sender to induce only the receiver's posterior beliefs that belong to the pool set $P$, creating combinatorial structures over the sender's optimal signaling policy.

The proof of Proposition~\ref{pro:rec-ut} leverages the geometry of convex functions to design the receiver's actions and utility values. Its main challenge is to use only polynomially-many receiver's actions to achieve the two desirable properties. We defer its technical argument to Appendix~\ref{append:prop:receiver-utility-design}. For the remainder of the proof, we shall utilize Properties~(1) and~(2) of Proposition~\ref{pro:rec-ut} only, rather than the actual construction in the proof.

\vspace{2mm}
\noindent   \textbf{Sender's utility $u_S$.}  For every receiver's action $a_{p,\omega_i}$ ($p\in P, \omega_i\in\Omega$), we set sender's utility $u_S\parentheses*{\omega,a_{p,\omega_i}}=-7$ at every state $\omega$. As mentioned previously, these are ``bad actions’’ for the sender, giving her the least possible utility regardless of the state. Now we turn to the actions $a_{p}$ for $p\in P$. We remind that for any state $\omega$, the statement $\omega \in \supp\parentheses*{p}$ is equivalent to saying that the state $\omega$ appears in the pool $p$ (i.e., a subset of states). We define the sender's utility by:

\begin{align*}
    u_S\parentheses*{\omega,a_p}=\begin{cases} 
1 &\text{ if } p\in P_v \text{ and } \omega\in \supp\parentheses*{p} \\
0 &\text{ if } p\in P_c\cup P_s \text{ and } \omega\in \supp\parentheses*{p} \\
-7 &\text{ otherwise.}
    \end{cases}
\end{align*}

Namely, the sender can gain positive utility only under some \emph{variable pool} $p \in P_v$, since this $p$ is the only posterior belief under which the action $a_p$ can be played. 
While aiming to put as high as possible mass on variable pools, in parallel the sender has the challenge of avoiding the high penalty of $-7$.  This is particularly challenging at a clause state $c_j$ -- unlike the variable states, which can form singleton pools in $P_s$, clause states must be pooled with some variable states to avoid the large penalty. However, this may contradict the sender's use of variable states to obtain positive utility under variable pools. Such conflicting reward-penalty tradeoff is the intrinsic source of hardness, as we shall formalize next. 

\subsubsection*{Step 2: Upper-Bounding Sender's Utility via Equilibrium Analysis}
The main goal of our second step is establishing an upper bound of $\frac{kd}{7m}$ on the sender's expected utility at any equilibrium. We achieve this by proving several lemmas on the structure of the sender-optimal equilibrium. These lemmas also build the connection between our constructed instance and the \texttt{Max-Var-3SAT} problem. 

Given any cheap talk equilibrium, a variable pool $p\in P_v$ will be called \emph{attractive} if the posterior $p$ is induced with positive probability by the sender, and under the posterior $p$ the receiver plays the action $a_p$ with a probability strictly above $ \frac{7}{8}$. Note that this probability yields a strictly positive expected sender's utility in the relevant states. We observe that in an equilibrium, if a variable pool $p\in P_v$ is attractive, then the sender sends the signal that induces the posterior $p$ at all relevant states (i.e., at states $x_{i,j}$ if $p=P_v\parentheses*{i}$ and at states $\neg x_{i,j}$ if $p=P_v\parentheses*{\neg i}$) \emph{with probability $1$}. This is because at these states, only the posterior $p$ yields a strictly positive sender's utility; this follows from our instance construction, where variable pools form a partition of all variable states, and a variable state can enjoy a positive utility if and only if the variable pool containing it is attractive.

Our first lemma states that there cannot be two contradicting attractive pools at any sender-optimal~equilibrium.

\begin{lemma}\label{lem:vcont}
 There is no sender-optimal equilibrium at which both pools $P_v\parentheses*{i}$ and $P_v\parentheses*{\neg i}$ are~attractive.
\end{lemma}

The proof appears in Appendix~\ref{append:prop:vcont}. Lemma~\ref{lem:vcont} implies that the attractive pools form a partial assignment for the variables. Formally, given an equilibrium, we define the \emph{induced (partial) assignment} as follows.

\begin{itemize}
    \item If $P_v\parentheses*{i}$ is an attractive pool, the induced assignment will be $x_i=$\texttt{False}.
    \item If $P_v\parentheses*{\neg i}$ is an attractive pool, the induced assignment will be $x_i=$\texttt{True}.
    \item If neither $P_v\parentheses*{i}$ nor $P_v\parentheses*{\neg i}$ is an attractive pool, we do not assign any value to $x_i$. 
\end{itemize}

Note that the fourth option of both being attractive is excluded by Lemma~\ref{lem:vcont}. The next lemma shows that in a sender-optimal equilibrium, contradictions cannot be created within a clause; for the proof, see Appendix~\ref{append:prop:ccont}.

\begin{lemma}\label{lem:ccont}
    For every sender-optimal equilibrium, its induced assignment to the variables of $\phi$ does not contradict any clause.  
\end{lemma}

Lemma~\ref{lem:ccont} implies that in a sender-optimal equilibrium, the induced partial assignment must not create contradictory clauses. Moreover, the sender's utility is at most $0$ for any posterior that is not an attractive pool. Therefore, we can bound from above the sender's utility in any equilibrium by $\frac{kd}{7m}$, where $k$ is the maximum size of a non-contradicting variable assignment. Note that $kd$ counts the total number of states appearing in all attractive pools, and $\frac{1}{7m}$ is a normalization factor due to the uniform prior over the $7m$~states.

\subsubsection*{Step 3: Lower-Bounding Sender's Utility via Equilibrium Construction} 
Finally, we construct an equilibrium that indeed achieves the $\frac{kd}{7m}$ expected sender's utility upper bound proved above. Given a partial assignment $x_S$ for $k$ variables with $\absolute*{S}=k$ and $x_i=$\texttt{True}/\texttt{False} for every $i\in S\subseteq \brackets*{n}$, we construct a cheap talk equilibrium as follows:

\begin{itemize}
    \item First, we create $k$ variable pools based on the partial assignment for the $k$ variables. If $x_i=$\texttt{True}, we create the pool $P_v\parentheses*{\neg x_i}$, whereas if $x_i=$\texttt{False} -- we create the pool $P_v\parentheses*{x_i}$.
    \item Second, for every clause $j\in \brackets*{m}$, the sender pools $c_j$ with the corresponding states $x_{i,j}$ or $\neg x_{i,j}$ in a way that is consistent with the assignment to create a clause pool for every clause. That is, if $x_i=$\texttt{True}, the sender pools $c_j$ with $x_{i,j}$ (since $\neg x_i$ was used to create the pool $P_v\parentheses*{\neg x_i}$ above); and if $x_i=$\texttt{False}, then the sender pools $c_j$ with $\neg x_{i,j}$. If $x_i$ is not assigned a value, then we include in $c_j$'s pool the state corresponding to that literal from the pair $\braces*{x_{i,j}, \neg x_{i,j}}$ which appears in the $j$th clause (we can consider this as allowing to give the variable $x_i$ different values to satisfy different clauses). In this way, we guarantee that each clause belongs to some clause pool.   
    \item Third, all the remaining variable states that have not been pooled yet are pooled into singleton pools. This completes the description of the sender's signaling policy.
    \item Finally, since the sender's signaling policy only leads to posteriors in the pool set $P$, the receiver's response is simply set to take action $a_p$ deterministically for every created pool $p$.
\end{itemize}

To see that the construction above is a cheap talk equilibrium, note that the receiver is obviously best responding by the construction of his utility as in Proposition~\ref{pro:rec-ut}. Moreover, the receiver's deduction of the posterior beliefs is Bayesian given the sender's strategy. To see that the sender cannot improve her utility by sending a different signal at any state, note that her utility at any state is either the largest possible utility $1$ (at a variable pool) or $0$ (at any other pools). Whenever her utility is $0$, she cannot induce a strictly positive utility by deviating to a different pool by our construction. Finally, note that the sender's utility at this equilibrium is precisely $\frac{kd}{7m}$, concluding Proposition~\ref{pro:reduction} proof.

\subsection{Extension to the Receiver's Perspective}
\label{sub:receiver}
We first show that Proposition~\ref{pro:signals} can be extended to show that there exists a receiver-optimal equilibrium using at most $n+1$ signals. Then we deduce from Proposition~\ref{pro:babbling} proof that it is NP-hard to decide whether the receiver can achieve a given level of utility at an equilibrium, and even to decide whether the babbling equilibrium is receiver-optimal.

\begin{proposition}
    \label{pro:receiver}
    For any expected receiver's utility value $v$ that arises at some cheap talk equilibrium, there exists a cheap talk equilibrium $\parentheses*{\pi,s}$ with $\absolute*{\supp\parentheses*{\pi}}\leq n+1$ and expected receiver's utility of $v$.
\end{proposition}

\begin{proof}
Fix an equilibrium $\parentheses*{\pi,s}$ with expected receiver's utility of $v$. As in Proposition~\ref{pro:signals} proof, we have $\mathbb{E}_{\sigma \sim \pi}\brackets*{p_\sigma}=\mu$. For any posterior $p_\sigma$, let $v\parentheses*{p_\sigma,s}$ be the corresponding expected receiver's utility. By Carathéodory's theorem applied on $\Delta\parentheses*{\Omega}\times \mathbb{R}$, there is a split $\nu$ of $\parentheses*{\mu,v}$ with support size at most $n+1$ that only uses posteriors from $\braces*{p_\sigma}_{\sigma \in \supp \parentheses*{\pi}}$. We have $\mathbb{E}_{\sigma \sim \nu}\brackets*{p_\sigma}=\mu$ and $\mathbb{E}_{\sigma \sim \nu}\brackets*{v\parentheses*{p_\sigma,s}}=v$. By the splitting lemma, we get that there exists a signaling policy $\pi'$ inducing a distribution over posteriors which is exactly $\nu$, having an expected receiver's utility of $v$ and using at most $n+1$ signals. Exactly as in Proposition~\ref{pro:signals} proof, we further get that $\parentheses*{\pi',s}$ is a cheap talk~equilibrium.
\end{proof}

Similarly to Proposition~\ref{pro:signals}, it follows from Proposition~\ref{pro:receiver} that the problem of finding a receiver-optimal cheap talk equilibrium belongs to NP. We show now that the following hardness result is implied by Proposition~\ref{pro:babbling} proof and a convexity argument.

\begin{proposition}
     \label{pro:receiver2}
     Deciding whether there is an equilibrium with expected receiver's utility higher than at a babbling equilibrium is NP-complete.
\end{proposition}
   
\begin{proof}
    We already know from Proposition~\ref{pro:receiver}  and the discussion after it that this problem belongs to NP. To show the hardness part, we first note that by~\cite{blackwell1953equivalent}, the receiver's utility as a function of the posterior, assuming that the receiver plays one of his expected utility-maximizing strategies over this posterior, is convex.

    In particular, it follows that at a babbling equilibrium, the receiver gets the minimum possible expected utility over all the cheap talk equilibria. Since determining whether a given 3CNF formula is satisfiable is NP-hard, it is enough to prove that the reduction in Proposition~\ref{pro:babbling} proof from Max-Var-3SAT to determining whether the \emph{sender} can achieve higher equilibrium utility than at a babbling equilibrium from Proposition~\ref{pro:babbling} proof satisfies the following property: The receiver's expected utility at the equilibrium $\parentheses*{\pi,s}$ constructed when the given 3CNF formula is satisfiable is strictly above the receiver's utility at a babbling equilibrium. Since the receiver's utility is convex as a function of the posterior, it suffices to prove that the expected receiver's utility of $\parentheses*{\pi,s}$ is not exactly as at a babbling equilibrium.

    Indeed, suppose that it does not hold. For any posterior $q$, let $v\parentheses*{q}$ be the maximal expected utility the receiver can achieve when the posterior is $q$ (where the expectation is over the posterior). We have $\mathbb{E}_{\sigma \sim \pi}\brackets*{p_\sigma}=\mu$ and $\mathbb{E}_{\sigma \sim \pi}\brackets*{v\parentheses*{p_\sigma}}=v\parentheses*{\mu}$. Therefore, the action $a_{\mu}$ provided by Proposition~\ref{pro:rec-ut}, which is a receiver's best-response to the posterior being the prior $\mu$, must also be a best response to each of the posteriors $p_\sigma\in\supp\parentheses*{\pi}$: Otherwise, we must have $\mathbb{E}_{\sigma \sim \pi}\brackets*{v\parentheses*{p_\sigma}}>v\parentheses*{\mu}$. This is a contradiction to the first condition of Proposition~\ref{pro:rec-ut}, which was applied on $\supp\parentheses*{\pi}\cup \braces*{\mu}$.
\end{proof}

\section{Tractable Cases}
\label{sec:tractble}
Following the negative result from the previous section, we turn to analyzing well-motivated special cases. We note that the previous intractability result holds when both the number of states and the number of actions are variable, and moreover -- sender's utility is state-dependent. In this section, we show that in the following two cases, computing a sender-optimal cheap talk equilibrium is tractable:  (i)~there are $n=O\parentheses*{1}$ states; (ii)~there are $m=2$ actions for the receiver.

\subsection{A Constant Number of States}
\label{sub:states}
In this subsection, we prove that finding a sender-optimal cheap talk equilibrium is computationally tractable when the number of states $n$ is constant. The proof is based on the bound of $n=O\parentheses*{1}$ on $\absolute*{\supp\parentheses*{\pi}}$ at a sender-optimal equilibrium, as ensured by Proposition~\ref{pro:signals}. This bound allows a translation of the cheap talk problem to a normal form game between the sender and the receiver in which the sender has $O\parentheses*{1}$ actions. The main result of this subsection is as follows.

\begin{proposition}
\label{pro:states}
    Suppose that the number of states $n$ is constant. Then one can compute a sender-optimal equilibrium $\parentheses*{\pi,s}$ in $\poly\parentheses*{m}$-time.
\end{proposition}

\begin{proof}
Fix a set of signals $\Sigma$ of size $n$. By Proposition~\ref{pro:signals}, this set is rich enough to include the support of the sender's signaling policy at some sender-optimal equilibrium. 

The cheap talk interaction can be viewed as a two-player interaction in which:

\begin{itemize}
    \item The sender has $n^n$ pure strategies (i.e., mappings from $\Omega$ to $\Sigma$), and the signaling policies are mixtures over these.
    \item The receiver has $m^n$ pure strategies (i.e., mappings from $\Sigma$ to $A$), and his mixed strategies are mixtures over these.
\end{itemize}

  Notice that the sender has $n^n=O\parentheses*{1}$ actions, and the receiver has $m^n=\poly\parentheses*{m}$ actions. Therefore, one can apply an exhaustive search over all supports (sets of pure strategies played with a positive probability) of size $n^n=O\parentheses*{1}$ for both players to find all the possible equilibrium outcomes (in terms of the two agents' utilities), including the sender-optimal one -- see, e.g., the \emph{support testing} algorithm in~\cite{von2007equilibrium}.\footnote{If the matrix of receiver's utilities does not have the full rank of $\min\braces*{m,n}$ then the exhaustive search might miss some equilibria, but it is still done over every equilibrium outcome in terms of utilities~\cite{von2007equilibrium}. More precisely, for any mixed Nash equilibrium, there exists an equilibrium yielding the same respective sender's and receiver's expected utilities s.t.~each player mixes over at most $n^n$ pure strategies.}
\end{proof}

\subsection{A Binary-Action Receiver}
\label{sub:binary}
In this subsection, we show that when the number of actions is $m=2$, finding a sender-optimal cheap talk equilibrium is computationally tractable, and it even can be done in time linear in $n$. Moreover, at most two signals suffice for the sender, with each signal being an incentive-compatible recommendation for the receiver to take a specific action. Namely, either a babbling equilibrium is optimal, or an optimal equilibrium is obtained when the sender greedily recommends the receiver to take the sender's preferred action in the current state, with ties broken in the receiver's favor.

\begin{definition}
    \label{def:greedy}
    A \emph{sender-greedy signaling policy} uses a signal per action -- the \emph{recommended action}. In each state, it sends a signal equal to a most preferred action by the sender in that state, with ties broken in the receiver's favor.
\end{definition}

\begin{proposition}
    \label{pro:binary}
    When the receiver has $m=2$ actions, there exists an algorithm computing a sender-optimal cheap talk equilibrium $\parentheses*{\pi,s}$ in time linear in $n$. Furthermore, either $\pi$ is a sender-greedy signaling policy and $s$ obeys the recommended actions, or $\parentheses*{\pi,s}$ is a babbling equilibrium.
\end{proposition}

The proof outline relies on the possibility of reducing the analysis of mixtures over two actions to a single-dimensional optimization problem (and, therefore, does not generalize to more than two actions). Specifically, let $\sigma_1$ be a signal s.t.~$s\brackets*{\sigma_1}$ assigns the highest probability to the receiver taking action $a_1$. Since there are just two actions, the sender weakly betters off deviating to deterministically transmitting $\sigma_1$ in any state in which she prefers action $a_1$ over $a_2$. Therefore, in all the states in which the sender strictly prefers $a_1$ over $a_2$, she must deterministically induce some posterior $p$ (the same for all these states); this posterior may further be induced in some states in which the sender is indifferent between the two actions. A simple analysis shows that if under $p$ the receiver weakly prefers $a_1$ over $a_2$, then he must deterministically play $a_1$ as a response to $p$ or we get a contradiction to sender-optimality;\footnote{Similar considerations can also be applied to the states in which the receiver prefers $a_2$ over $a_1$.} and if it is not possible to induce such $p$, then the sender cannot do better than a babbling equilibrium. For the formal proof, see Appendix~\ref{append:tractable:pro-binary}.

\section{Discussion}
\label{sec:discussion}
\paragraph{Maximization objective.} We have mainly focused on maximizing the sender's utility in equilibrium. Maximizing the receiver's utility has also been considered (Subsection~\ref{sub:receiver}). While for maximizing the sender's utility we have proved the hardness of a constant approximation, for the receiver we have shown only the hardness of computing the exact maximum. The hardness of approximating the maximal receiver's utility (or conversely a polynomial approximation algorithm) remains an open problem. 

Another possible objective is \emph{social welfare}, which is defined as the sum of the sender's and receiver's utilities $u_S+u_R$. It immediately follows from Theorem~\ref{thm:additive} that computing a certain additive or multiplicative constant approximation for the social welfare-maximizing equilibrium is NP-hard. This is because one can scale down the receiver's utility function to make it negligible (e.g., smaller than $c/2$ in Theorem~\ref{thm:additive} notations) compared to the sender's utility. Moreover, note that the proof outline of Proposition~\ref{pro:receiver} works for any linear combination of the sender's and the receiver's utilities instead of the receiver's utility $v$. In particular, there always exists a social welfare-maximizing equilibrium using at most $n+1$ signals. Therefore, the problem of deciding whether the social welfare can reach a certain threshold at an equilibrium belongs to NP. We have yet to find a positive computational result on maximizing social welfare or the receiver's utility in a cheap talk equilibrium. 

\vspace{2mm}
\noindent \textbf{Further research on algorithmic cheap talk.} Our results lead to many interesting opportunities for future research. An immediate question following our hardness result is whether there is a quasi-polynomial time algorithm for finding a sender-optimal or welfare-maximizing cheap talk equilibrium, and what the tight inapproximability constant $c$ is.

Given the rich body of algorithmic studies of Bayesian persuasion, it is not difficult to see the various possible follow-up algorithmic directions in extensions of the cheap talk model. Examples include introducing many receivers, communicating with the sender privately or publicly, with or without receivers' externalities. Another direction involves studying approximate equilibria in cheap talk (cf.~\cite{rubinstein2016settling}) and the complexity of finding an optimal such solution. Furthermore, we prove that when there are constantly-many states, it is possible to compute a sender-optimal equilibrium in a polynomial time. However, it is an interesting open question on whether an analogous result holds for constantly many actions, even for the constant being three. It will also be interesting to investigate computational applications of cheap talk in domains such as entertainment games with communications~\cite{meta2022human}, security games and negotiations, in which commitment is often impossible.

\vspace{2mm}
\noindent \textbf{Alternative frameworks for strategic information transmission.}  Another interesting research direction is analyzing the computational aspects of intermediate models between the two extreme settings -- Bayesian persuasion (full sender's commitment power) and cheap talk (no commitment at all). A natural question is whether the sender's optimization problem is tractable for persuasion with partial sender's commitment -- e.g., in the credible persuasion framework~\cite{lin2022credible}.

\paragraph{Acknowledgements.} Yakov Babichenko is supported by the Binational Science Foundation BSF grant No.~2018397. This work is funded by the European Union (ERC, ALGOCONTRACT, 101077862, PI: Inbal Talgam-Cohen). Inbal Talgam-Cohen is supported by a Google Research Scholar award and by the Israel Science Foundation grant No.~336/18. Yakov and Inbal are supported by the Binational Science Foundation BSF grant No.~2021680. Haifeng Xu is supported by NSF award No.~CCF-2303372, Army Research Office Award No.~W911NF-23-1-0030 and Office of Naval Research Award No.~N00014-23-1-2802. Konstantin Zabarnyi is supported by a PBC scholarship for Ph.D. students in data science. The authors are grateful to Mirela Ben-Chen, Dirk Bergemann, Yang Cai, Françoise Forges, Elliot Lipnowski and anonymous reviewers for helping them to improve the paper.

\bibliographystyle{plainnat}
\bibliography{main}

\appendix

\section{Omitted Proofs in Section \ref{sec:hard}}
\subsection{Proof of Proposition~\ref{pro:babbling}}
\label{append:prop:babbling}
The reduction in Theorem~\ref{thm:additive} proof constructs cheap talk instances in which it is NP-hard to decide whether there exists a cheap talk equilibrium with an expected utility of at least some $\beta\in\parentheses*{0,1}$ for the sender or there is no cheap talk equilibrium with an expected utility of at most $\alpha$ for some $\alpha<\beta$. Moreover, the cheap talk equilibrium $\parentheses*{\pi,s}$ with expected sender's utility of at least $\beta$ constructed in Theorem~\ref{thm:additive} proof does not induce a posterior equal to the prior with a positive probability (i.e., $p_{\sigma}\neq \mu$ for any $\sigma \in \supp\parentheses*{\pi}$).

We modify our reduction as follows. We add the prior $\mu$ to the set of special posteriors (i.e., pools). We further add an action $a_0$ for the receiver that is a best reply only when $p_{\sigma}=\mu$; namely, we apply Proposition~\ref{pro:rec-ut} with one additional point $\mu$. We set the sender's utility for the action $a_0$ being state-independent and equal to $\alpha$: $u_S\parentheses*{a_0}=\alpha$. After the modification, the sender's utility in the sender-optimal babbling equilibrium is $\alpha$.

If the original instance has an equilibrium with expected sender's utility of at least $\beta$, the same must be true for the new instance, as we know from Theorem~\ref{thm:additive} proof that there must exist such an equilibrium with no signal inducing a posterior equal to $\mu$. Therefore, in such a case, there exists an equilibrium with sender's expected utility strictly above $\alpha$ in the modified instance.

Assume now that in the original instance, all equilibria have expected sender's utility of at most $\beta$. Then no equilibrium in the modified instance can have expected sender's utility larger than $\alpha$ without inducing a posterior equal to the prior with a positive probability -- if no such posterior is induced, then the new action $a_0$ is irrelevant. It remains to consider equilibria $\parentheses*{\pi,s}$ in the modified instance with some $\sigma \in \supp\parentheses*{\pi}$ satisfying $p_{\sigma}=\mu$. For such $\sigma$, the sender gets utility of at most $\alpha$. Moreover, $\sigma$ is sent in every state with a positive probability by its definition. We get from the sender's indifference constraints that in every state, the sender's expected utility over $\pi$ is at most $\alpha$. In any case, the sender's expected utility in any equilibrium is at most $\alpha$, as desired.  

\subsection{Proof of Proposition~\ref{pro:regular}}\label{append:prop:sat-hard}
Our reduction starts from the following known hard problem.

\begin{lemma}\label{lem:gap-sat}[adapted from Theorem~2 of~\cite{berman2004approximation}]
For every $\epsilon \in \parentheses*{0, 1/2}$, there exists a family of $4$-regular $3$SAT instances such that each instance $\phi$ consists of $m = 1016M$ clauses for some $M$, and it is NP-hard to decide whether $\phi$, which is promised to be in one of the following two situations, is actually in the first of them:
 
  (\textbf{A1})  at least  $\parentheses*{1016 - \epsilon}M$ clauses of $\phi$ can be satisfied  simultaneously;
 
   (\textbf{A2}) at most $\parentheses*{1015 + \epsilon}M$ clauses in  $\phi$ can be satisfied  simultaneously. 
\end{lemma} 

To avoid carrying cumbersome numbers, let $q_1 := \frac{1016 - \epsilon}{1016}$ be the fraction of the $m=1016M$ clauses that case \textbf{A1} can satisfy, whereas $q_2 := \frac{  1015 + \epsilon}{1016}$ is the fraction for case \textbf{A2}. Note that $3m = 4n$. We reduce the problem in Lemma~\ref{lem:gap-sat} to distinguishing between the following two situations for our 4-regular \texttt{Max-Var-3SAT} problem, when promised that $\phi$ is in one of the following two~cases: 
 
  (\textbf{B1})  $\phi$ admits a non-contradictory partial assignment to at least $\frac{4 q_1 -1}{3}n = \frac{3048 - 4 \epsilon}{3048} n $ variables; 
 
   (\textbf{B2}) any partial assignment to strictly more than  $\frac{2 + q_2}{3}n =\frac{3047 +  \epsilon}{3048} n  $ variables must lead to contradictory clauses in~$\phi$.

The proposition proof follows by taking $\epsilon=0.1$. The reduction goes by proving that, for any $4$-regular $3$SAT formula $\phi$, case \textbf{A1} implies case \textbf{B1} and case \textbf{A2} implies case \textbf{B2}. Consequently, if there is an algorithm for solving the above distinguishing problem for 4-regular \texttt{Max-Var-3SAT} problem, this algorithm -- when used for the instances from Lemma \ref{lem:gap-sat} --  can also distinguish case \textbf{A1} (which satisfies \textbf{B1}) from case \textbf{A2} (which satisfies \textbf{B2}). It, thus, implies the NP-hardness of distinguishing \textbf{B1} from \textbf{B2} for 4-regular \texttt{Max-Var-3SAT} instances, when promised that the instance is at one of the two cases. 

We first argue that case \textbf{A1} implies case \textbf{B1}. Suppose that there is a full variable assignment that can satisfy at least $q_1m$ clauses. This means there are at most $\parentheses*{1-q_1}m$ unsatisfied clauses. We can modify this full assignment by removing any one variable from each of the unsatisfied clauses and make its value unassigned. This results in a partial assignment with at least $n - \parentheses*{1-q_1}m = n-\frac{4\parentheses*{1-q_1}n}{3}=\frac{4q_1 - 1}{3}n $ assigned variables. Therefore, if the  instance $\phi$ belongs to case \textbf{A1}, it must belong to  case \textbf{B1}. 

Next, we argue that case \textbf{A2} implies case \textbf{B2}. We prove its contrapositive -- that is, if there is a non-contradictory partial assignment of size strictly larger than  $\frac{2 + q_2}{3}n $, then strictly more than $q_2m$ clauses can be satisfied in $\phi$. Indeed, this non-contradictory partial assignment leaves less than $\frac{1-q_2}{3}n$ variables unassigned, and these variables can span less than $\frac{4(1-q_2)}{3}n$ clauses. Therefore, more than  $m - \frac{4(1-q_2)}{3}n = q_2m$ clauses are satisfied by this partial assignment, as desired.

\subsection{Proof of Proposition \ref{pro:rec-ut}}\label{append:prop:receiver-utility-design}

 \begin{figure}[h] 
		\centering
		\includegraphics[width=7cm]{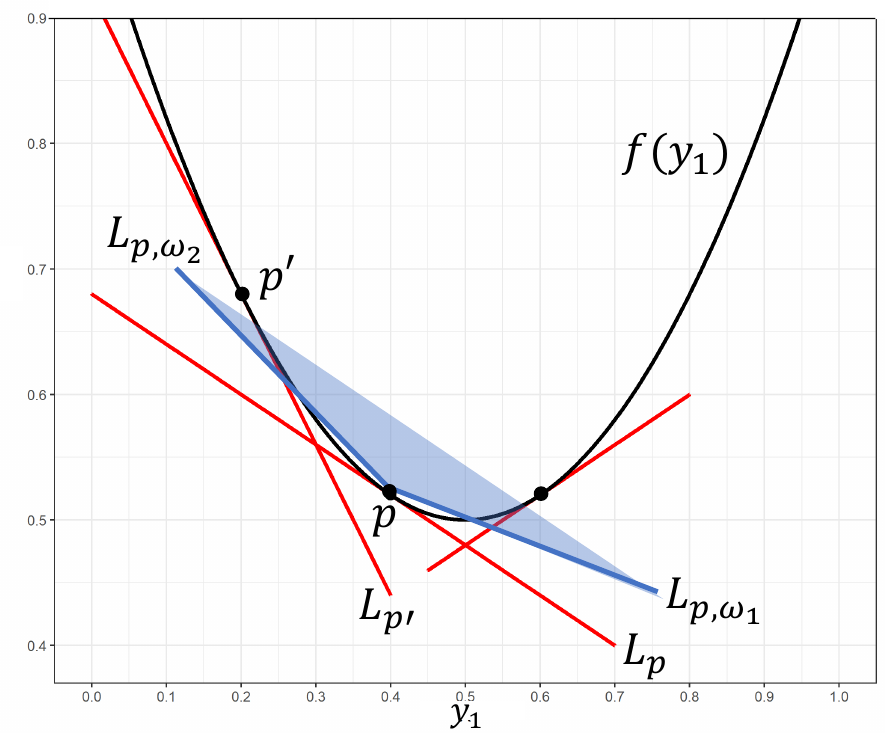}
		\caption{Geometric Illustration for the Proof of Proposition \ref{pro:rec-ut}.}
		\label{fig:GeoReceiverU}
	\end{figure} 
 
   The proof is geometric. We start with considering some strictly convex function $f:\Delta\parentheses*{\Omega}\to \mathbb{R}$. For instance, one can pick $f\parentheses*{y}=\sum_i y_i^2$, though any strictly convex function will be equally good. An illustration for the two-state situation is provided in Figure~\ref{fig:GeoReceiverU}, where the input variable $y_1$ to $f\parentheses*{y_1}$ denotes the probability of the first state. 

   For every $p\in P$, let $L_p: \Delta\parentheses*{\Omega}\to \mathbb{R}$ be the tangent plane to the plot of $f$ at the point $p$. Notice that $L_p$ uniquely defines an action $a_p$ for the receiver in terms of receiver's utility. Indeed, given some state $\omega_i$ ($i\in\brackets*{n}$), consider the vector $e_i$ of length $n$ with the $i$th coordinate being $1$ and the other coordinates being $0$. Setting $u_R\parentheses*{\omega_i, a_p}=L_p\parentheses*{e_i}$ ensures, by linearity of receiver's expected utility as a function of the posterior $y\in \Delta\parentheses*{\Omega}$, that $L_p\parentheses*{y}$ exactly equals the expected receiver's utility at the posterior $y$.
   
   Fix $p'\in\Delta\parentheses*{\Omega}\setminus p$. The strict convexity of $f$ implies that $L_{p'}\parentheses*{y}<f\parentheses*{y}$ for every $y\neq p'$. In particular, for $y=p\neq p'$,  we  get  $L_{p'}\parentheses*{p}<f\parentheses*{p}=L_{p}\parentheses*{p}$, where $L_{p'}\parentheses*{p}$ is receiver's utility of action $a_{p'}$ at  the posterior belief $p$. This implies that $a_{p}$ is the unique best response among the actions $\{a_{p'}\}_{p'\in P}$ at the posterior $p$, which ensures Property~(1).

   Now we turn to define the additional actions $a_{p,\omega}$ to ensure Property~(2). Let us set\\$\epsilon:=\min_{p, p' \in P: p\neq p'} \braces*{L_p\parentheses*{p}-L_{p'}\parentheses*{p}}$. Since the set $P$ is finite, $\epsilon$ is well-defined and strictly positive. We want to make the action $a_p$ sub-optimal for every $y\neq p$. Geometrically, we create a reverse pyramid whose ``button vertex’’ is $p$ (the shadow region of Figure~\ref{fig:GeoReceiverU}). The facets of the pyramid correspond to these additional actions. However, additional carefulness is needed when defining this pyramid. We must avoid the situations in which one the new facets will strictly exceed the value of $L_{p'}\parentheses*{p'}$ for some point $p'\neq p$, which will ruin Property~(1). Under the illustration of Figure~\ref{fig:GeoReceiverU}, we must have $L_{p, \omega} \parentheses*{p'} \leq L_{p'} \parentheses*{p'}$ for every $\omega$ and $p'$. It turns out that such actions can be explicitly constructed.

   Indeed, for every $p=\parentheses*{p_{\omega_i}}_{i\in\brackets*{n}}\in P$ and $\omega \in \Omega$, let us define the linear function $L'_{p,\omega}\parentheses*{y}:=y_{\omega}-p_{\omega}$. Note that $L'_{p,\omega}\parentheses*{y}\leq 1$ and $L'_{p,\omega}\parentheses*{p}=0$. Moreover,  for every $y\neq p$, there exists $\omega \in \Omega$ such that $L'_{p,\omega}\parentheses*{y}>0$. This is because $\sum_{\omega \in \Omega} L'_{p,\omega}(y) = \sum_{\omega \in \Omega}  \parentheses*{y_{\omega}-p_{\omega}} = 1 - 1 = 0$, and thus we know that at least one of $\braces*{L'_{p,\omega}\parentheses*{y}}_{\omega}$ is strictly positive when $y\neq p$.   
   
   We now define the linear function $L_{p,\omega}:=L_p+\epsilon L'_{p,\omega}$. For a fixed $p$ and variable values of $\omega$, these linear functions can be thought of as pyramid facets (see Figure~\ref{fig:GeoReceiverU}). Denote the action that corresponds to some $L_{p,\omega}$ by $a_{p,\omega}$ -- i.e., $u_R\parentheses*{\omega,a_{p,\omega_i}}=L_{p,\omega_i}\parentheses*{e_i}$ ($i\in\brackets*{n}$). We argue that these additional actions $a_{p,\omega}$ ensure Property~(2) without violating Property~(1).

To see that Property~(2) holds, we observe that for every $y\notin P$ and every $p\in P$, we have $L_{p,\omega}\parentheses*{y}=L_p\parentheses*{y}+\epsilon L'_{p,\omega}\parentheses*{y}>L_p\parentheses*{y}$ for $\omega\in \Omega$ satisfying $L'_{p,\omega}\parentheses*{y}>0$. Namely, any $a_p$ is not a best response under any posterior belief $y \notin P$.  To see that Property~(1) is not violated, we observe the following two facts:

\begin{itemize}
    \item For any $p' \not = p$, we have $L_{p',\omega}(p)= L_{p'}\parentheses*{p}+\epsilon L'_{p',\omega}\parentheses*{p}\leq L_{p'}\parentheses*{p}+\epsilon \leq L_p\parentheses*{p}$, where the last inequality follows from the definition of $\epsilon$.
    \item Additionally,  we have $L_{p,\omega}\parentheses*{p}= L_{p}\parentheses*{p}+\epsilon L'_{p,\omega}\parentheses*{p} = L_p\parentheses*{p}$.
\end{itemize}

Thus, even with introduction of the $a_{p', \omega}$ actions,  $a_p$ remains a best response under the posterior $p$ (though it is not the unique best response, since all the $a_{p, \omega}$ actions are also best responses). This concludes the proof.

\subsection{Proof of Lemma~\ref{lem:vcont}}
\label{append:prop:vcont}
We shall show that any cheap talk equilibrium containing contradicting attractive pools cannot be sender-optimal. Specifically, we can adjust it to a strictly better equilibrium for the sender. 
 
Let $C\subseteq \brackets*{n}$ denote the set of indices of contradicting attractive pools (i.e., $i \in C$ means both $P_v\parentheses*{i}$ and $P_v\parentheses*{\neg i}$ are attractive). For every clause $j$ that contains a variable $x_i$ from $i\in C$, we know that the sender's utility at the state $c_j$ is necessarily $-7$. This is because the only way to avoid the penalty of $-7$ is by pooling $c_j$ into a clause pool; however, since in both states  $x_{i,j}$ and $\neg x_{i,j}$ the sender sends their corresponding variable pool signal $P_v\parentheses*{i}, P_v\parentheses*{\neg i}$ with probability $1$, the mass of these states is taken and cannot be pooled with $c_j$.

We can apply the following modification to avoid this penalty (see Figure~\ref{fig:vcont}  for illustration). First, we break all the contradicting attractive pools in $C$. Second, for every variable state that currently is not part of any attractive variable pool,  we separate out all its probability mass from the posteriors it belongs to. Note that after these two modifications,  for every variable $x_{i}$ (not necessarily in $C$), at least one of $x_{i, j}, \neg x_{i, j}$ is fully available in the sense that its probability mass is not used in any of the remaining pools or posteriors. Third, for every clause $j$ that contains a variable $x_i$ from $i\in C$, we pool every clause state $c_j$ with the currently available variable states into a clause pool. Since we have broken both attractive variable pools $P_v\parentheses*{i}$ and $P_v\parentheses*{\neg i}$, at least one of them can be pooled with $c_j$, together with the other two available variables (as for at least one of them we shall, indeed, get a clause pool). Fourth, all the remaining variable states that have not been pooled after this operation are set as singleton pools. Finally, we adjust the receiver's strategy to best respond to all the new pools, and in the sender's optimal manner when there is indifference.

To see that the modifications above still result in a cheap talk equilibrium, note that all the variable states in the modified strategies are deterministically set to one of the following three types of signals: (1)~an attractive variable pool; (2)~a singleton variable pool consisting of the variable state itself; (3)~a clause pool containing this variable. Thus, the remaining posteriors can only be supported on some clause states that do not appear in any of the clause pools, and they will lead to the sender's utility of $-7$ under \emph{any} signal. By construction, the sender's strategy after the above modification is a best response to the receiver's strategy.

\begin{figure}[ht]
     \centering
     \begin{tikzpicture}[scale=0.32]
\draw[rounded corners] (0, 0) rectangle (8, -1) {};
\node at (3.5,-0.5) {$x_{1,1}$};
\draw[rounded corners] (9, 0) rectangle (17, -1) {};
\node at (12.5,-0.5) {$x_{1,2}$};

\draw[blue, rounded corners]  (-3.3,0.2) rectangle (17.2,-1.2) {};
\node at (-1.7,-0.5) {$P_v(1)$};

\draw[rounded corners] (0, -2) rectangle (8, -3) {};
\node at (3.5,-2.5) {$\neg x_{1,1}$};
\draw[rounded corners] (9, -2) rectangle (17, -3) {};
\node at (12.5,-2.5) {$\neg x_{1,2}$};

\draw[blue, rounded corners]  (-3.3,-1.8) rectangle (17.2,-3.2) {};
\node at (-1.7,-2.5) {$P_v(\neg 1)$};

\draw[rounded corners] (0, -4) rectangle (8, -5) {};
\node at (3.5,-4.5) {$x_{2,1}$};
\draw[rounded corners] (9, -4) rectangle (17, -5) {};
\node at (12.5,-4.5) {$x_{2,2}$};

\draw[blue, rounded corners]  (-3.3,-3.8) rectangle (17.2,-5.2) {};
\node at (-1.7,-4.5) {$P_v(2)$};

\draw[rounded corners] (0, -6) rectangle (8, -7) {};
\node at (3.5,-6.5) {$\neg x_{2,1}$};
\draw[rounded corners] (9, -6) rectangle (17, -7) {};
\node at (12.5,-6.5) {$\neg x_{2,2}$};

\draw[blue, rounded corners]  (-3.3,-5.8) rectangle (17.2,-7.2) {};
\node at (-1.7,-6.5) {$P_v(\neg 2)$};

\draw[rounded corners] (0, -8) rectangle (8, -9) {};
\node at (3.5,-8.5) {$x_{3,1}$};

\draw[rounded corners] (0, -10) rectangle (8, -11) {};
\node at (3.5,-10.5) {$\neg x_{3,1}$};

\draw[rounded corners] (9, -12) rectangle (17, -13) {};
\node at (12.5,-12.5) {$x_{4,2}$};

\draw[rounded corners] (9, -14) rectangle (17, -15) {};
\node at (12.5,-14.5) {$\neg x_{4,2}$};

\draw[rounded corners] (0, -16) rectangle (8, -17) {};
\node at (3.5,-16.5) {$c_1$};
\draw[rounded corners] (9, -16) rectangle (17, -17) {};
\node at (12.5,-16.5) {$c_2$};

\node at (7.5,-19.5) {Pre-modification};

\filldraw (6.6,-2.5) circle(0.2);
\filldraw (7.5,-2.5) circle(0.2);

\filldraw (0.5,-0.5) circle(0.2);
\filldraw (0.5,-4.5) circle(0.2);
\filldraw (0.5,-8.5) circle(0.2);
\filldraw (0.5,-16.5) circle(0.2);

\filldraw (1.4,-0.5) circle(0.2);
\filldraw (1.4,-4.5) circle(0.2);
\filldraw (1.4,-10.5) circle(0.2);
\filldraw (1.4,-16.5) circle(0.2);

\filldraw (2.3,-0.5) circle(0.2);
\filldraw (2.3,-6.5) circle(0.2);
\filldraw (2.3,-8.5) circle(0.2);
\filldraw (2.3,-16.5) circle(0.2);

\filldraw (4.8,-0.5) circle(0.2);
\filldraw (4.8,-6.5) circle(0.2);
\filldraw (4.8,-10.5) circle(0.2);
\filldraw (4.8,-16.5) circle(0.2);

\filldraw (5.7,-2.5) circle(0.2);
\filldraw (5.7,-4.5) circle(0.2);
\filldraw (5.7,-8.5) circle(0.2);
\filldraw (5.7,-16.5) circle(0.2);

\filldraw (6.6,-2.5) circle(0.2);
\filldraw (6.6,-4.5) circle(0.2);
\filldraw (6.6,-10.5) circle(0.2);
\filldraw (6.6,-16.5) circle(0.2);

\filldraw (7.5,-2.5) circle(0.2);
\filldraw (7.5,-6.5) circle(0.2);
\filldraw (7.5,-10.5) circle(0.2);
\filldraw (7.5,-16.5) circle(0.2);

\filldraw (9.5,-0.5) circle(0.2);
\filldraw (9.5,-4.5) circle(0.2);
\filldraw (9.5,-12.5) circle(0.2);
\filldraw (9.5,-16.5) circle(0.2);

\filldraw (10.4,-0.5) circle(0.2);
\filldraw (10.4,-4.5) circle(0.2);
\filldraw (10.4,-14.5) circle(0.2);
\filldraw (10.4,-16.5) circle(0.2);

\filldraw (11.3,-0.5) circle(0.2);
\filldraw (11.3,-6.5) circle(0.2);
\filldraw (11.3,-12.5) circle(0.2);
\filldraw (11.3,-16.5) circle(0.2);

\filldraw (13.8,-2.5) circle(0.2);
\filldraw (13.8,-4.5) circle(0.2);
\filldraw (13.8,-12.5) circle(0.2);
\filldraw (13.8,-16.5) circle(0.2);

\filldraw (14.7,-2.5) circle(0.2);
\filldraw (14.7,-4.5) circle(0.2);
\filldraw (14.7,-14.5) circle(0.2);
\filldraw (14.7,-16.5) circle(0.2);

\filldraw (15.6,-2.5) circle(0.2);
\filldraw (15.6,-6.5) circle(0.2);
\filldraw (15.6,-12.5) circle(0.2);
\filldraw (15.6,-16.5) circle(0.2);

\filldraw (16.5,-2.5) circle(0.2);
\filldraw (16.5,-6.5) circle(0.2);
\filldraw (16.5,-14.5) circle(0.2);
\filldraw (16.5,-16.5) circle(0.2);
     \end{tikzpicture}
    \hspace{14mm}
     \begin{tikzpicture}[scale=0.32]
\draw[rounded corners] (0, 0) rectangle (8, -1) {};
\node at (3.5,-0.5) {$x_{1,1}$};
\draw[rounded corners] (9, 0) rectangle (17, -1) {};
\node at (12.5,-0.5) {$x_{1,2}$};

\draw[rounded corners] (0, -2) rectangle (8, -3) {};
\node at (3.5,-2.5) {$\neg x_{1,1}$};
\draw[rounded corners] (9, -2) rectangle (17, -3) {};
\node at (12.5,-2.5) {$\neg x_{1,2}$};

\draw[red, rounded corners] (-0.2, -1.8) rectangle (8.2, -3.2) {};
\draw[red, rounded corners] (8.8, -1.8) rectangle (17.2, -3.2) {};

\draw[rounded corners] (0, -4) rectangle (8, -5) {};
\node at (3.5,-4.5) {$x_{2,1}$};
\draw[rounded corners] (9, -4) rectangle (17, -5) {};
\node at (12.5,-4.5) {$x_{2,2}$};

\draw[red, rounded corners] (-0.2, -3.8) rectangle (8.2, -5.2) {};
\draw[red, rounded corners] (8.8, -3.8) rectangle (17.2, -5.2) {};

\draw[rounded corners] (0, -6) rectangle (8, -7) {};
\node at (3.5,-6.5) {$\neg x_{2,1}$};
\draw[rounded corners] (9, -6) rectangle (17, -7) {};
\node at (12.5,-6.5) {$\neg x_{2,2}$};

\draw[rounded corners] (0, -8) rectangle (8, -9) {};
\node at (3.5,-8.5) {$x_{3,1}$};

\draw[rounded corners] (0, -10) rectangle (8, -11) {};
\node at (3.5,-10.5) {$\neg x_{3,1}$};

\draw[rounded corners] (9, -12) rectangle (17, -13) {};
\node at (12.5,-12.5) {$x_{4,2}$};

\draw[rounded corners] (9, -14) rectangle (17, -15) {};
\node at (12.5,-14.5) {$\neg x_{4,2}$};

\draw[rounded corners] (0, -16) rectangle (8, -17) {};
\node at (3.5,-16.5) {$c_1$};
\draw[rounded corners] (9, -16) rectangle (17, -17) {};
\node at (12.5,-16.5) {$c_2$};

\node at (7.5,-19.5) {Post-modification};

\filldraw (6.6,-2.5) circle(0.2);
\filldraw (7.5,-2.5) circle(0.2);

\filldraw (0.5,-0.5) circle(0.2);
\filldraw (0.5,-4.5) circle(0.2);
\filldraw (0.5,-8.5) circle(0.2);
\filldraw (0.5,-16.5) circle(0.2);

\filldraw (1.4,-0.5) circle(0.2);
\filldraw (1.4,-4.5) circle(0.2);
\filldraw (1.4,-10.5) circle(0.2);
\filldraw (1.4,-16.5) circle(0.2);

\filldraw (2.3,-0.5) circle(0.2);
\filldraw (2.3,-6.5) circle(0.2);
\filldraw (2.3,-8.5) circle(0.2);
\filldraw (2.3,-16.5) circle(0.2);
\draw (2.3,-0.5)--(2.3,-16.5);

\filldraw (4.8,-0.5) circle(0.2);
\filldraw (4.8,-6.5) circle(0.2);
\filldraw (4.8,-10.5) circle(0.2);
\filldraw (4.8,-16.5) circle(0.2);

\filldraw (5.7,-2.5) circle(0.2);
\filldraw (5.7,-4.5) circle(0.2);
\filldraw (5.7,-8.5) circle(0.2);
\filldraw (5.7,-16.5) circle(0.2);

\filldraw (6.6,-2.5) circle(0.2);
\filldraw (6.6,-4.5) circle(0.2);
\filldraw (6.6,-10.5) circle(0.2);
\filldraw (6.6,-16.5) circle(0.2);

\filldraw (7.5,-2.5) circle(0.2);
\filldraw (7.5,-6.5) circle(0.2);
\filldraw (7.5,-10.5) circle(0.2);
\filldraw (7.5,-16.5) circle(0.2);

\filldraw (9.5,-0.5) circle(0.2);
\filldraw (9.5,-4.5) circle(0.2);
\filldraw (9.5,-12.5) circle(0.2);
\filldraw (9.5,-16.5) circle(0.2);

\filldraw (10.4,-0.5) circle(0.2);
\filldraw (10.4,-4.5) circle(0.2);
\filldraw (10.4,-14.5) circle(0.2);
\filldraw (10.4,-16.5) circle(0.2);

\filldraw (11.3,-0.5) circle(0.2);
\filldraw (11.3,-6.5) circle(0.2);
\filldraw (11.3,-12.5) circle(0.2);
\filldraw (11.3,-16.5) circle(0.2);
\draw (11.3,-0.5)--(11.3,-16.5);

\filldraw (13.8,-2.5) circle(0.2);
\filldraw (13.8,-4.5) circle(0.2);
\filldraw (13.8,-12.5) circle(0.2);
\filldraw (13.8,-16.5) circle(0.2);

\filldraw (14.7,-2.5) circle(0.2);
\filldraw (14.7,-4.5) circle(0.2);
\filldraw (14.7,-14.5) circle(0.2);
\filldraw (14.7,-16.5) circle(0.2);

\filldraw (15.6,-2.5) circle(0.2);
\filldraw (15.6,-6.5) circle(0.2);
\filldraw (15.6,-12.5) circle(0.2);
\filldraw (15.6,-16.5) circle(0.2);

\filldraw (16.5,-2.5) circle(0.2);
\filldraw (16.5,-6.5) circle(0.2);
\filldraw (16.5,-14.5) circle(0.2);
\filldraw (16.5,-16.5) circle(0.2);

     \end{tikzpicture}
     \caption{Illustration of a possible re-pooling in Lemma~\ref{lem:vcont} proof for $n=4$, $m=2$, with the clauses being $c_1 = x_1 \lor x_2 \lor \neg x_3$ and $c_2 = \neg x_1 \lor x_2 \lor x_4$. The black rectangles represent states; the blue rectangles -- attractive variable pools (note that $P_v(1)$ and $P_v(\neg 1)$ are contradicting, and so are $P_v(2)$ and $P_v(\neg 2)$); the black lines -- clause pools (with dots marking the states in these pools); and the red rectangles -- singleton pools. The modification breaks all contradicting attractive variable pools and creates new $2$ clause and $4$ singleton pools. The clause pools correspond to the assignment $\parentheses*{x_1,x_2,x_3,x_4}=\parentheses{\texttt{True},\texttt{False},\texttt{True},\texttt{True}}$ that satisfies both clauses. The signaling policy in states $\neg x_{3,1}$ and $\neg x_{4,2}$ remains unchanged during the~modification.}
	\label{fig:vcont}
 \end{figure}
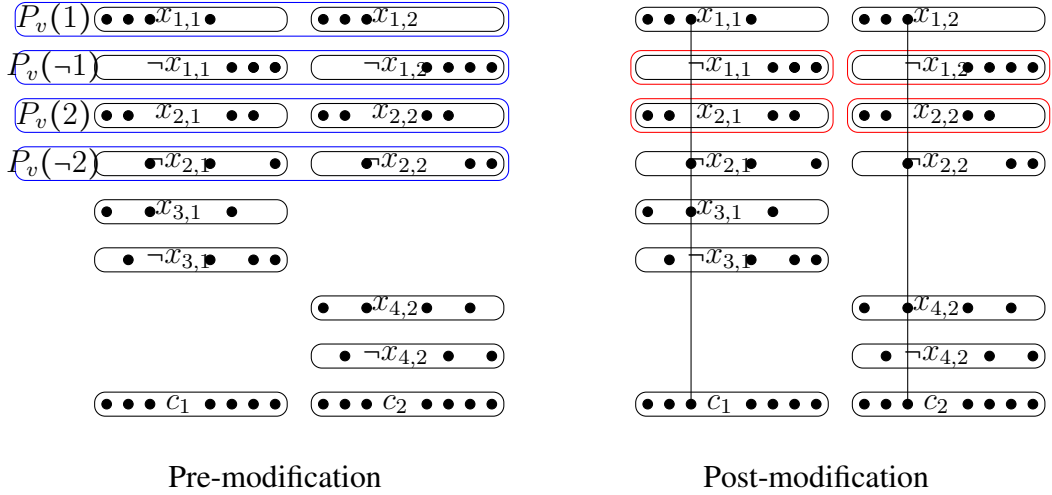

 Finally, let us count the sender's gains and losses from the above modification.  For every clause $j$ that contains a  variable $x_i$ for some $i\in C$, we have increased the utility at clause state $c_j$ by $7$. The total number of such clauses is at least $|C|d/3$. Meanwhile, there are $|C|d \times 2$ states at which we have decreased the utility by $1$ due to breaking all the contradicting attractive variable pools in $C$. In all other states, the utility remained unchanged. Since $7 \times \frac{|C|d}{3} - |C|d \times 2 > 0$,  such a modification strictly improves the sender's~utility.

\subsection{Proof of Lemma~\ref{lem:ccont}}
\label{append:prop:ccont}
Given an equilibrium that contradicts a clause, we modify it to a better one for the sender. Let $j\in \brackets*{m}$ be a contradictory clause that contains the variables $x_{i_1},x_{i_2},x_{i_3}$. For simplicity of notations, we assume that clause $j$ is $x_{i_1} \vee x_{i_2} \vee x_{i_3}$ and the induced assignment is $x_{i_1}=x_{i_2}=x_{i_3}=\texttt{False}$, which leads to contradiction; same arguments apply for the other seven cases. This means that $P_v\parentheses*{i_1}, P_v\parentheses*{i_2}$ and $P_v\parentheses*{i_3}$ are attractive variable pools, and that the pool that consists of the states $\braces*{c_j , \neg x_{i_1,j}, \neg x_{i_2,j}, \neg x_{i_3,j}}$ is not one of the clause pools of the $j$th clause.

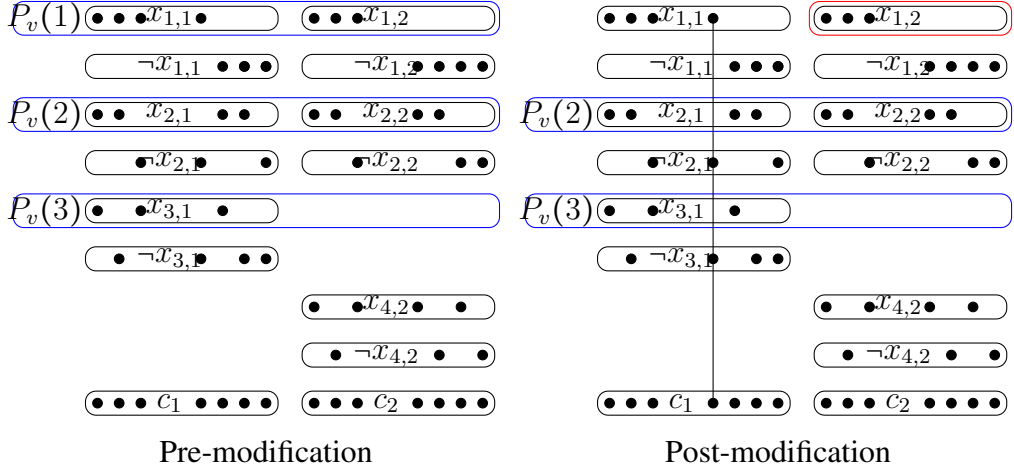
\begin{figure}[ht]
     \centering
     \begin{tikzpicture}[scale=0.32]
\node at (7.5,-18.5) {Pre-modification};
     
\draw[rounded corners] (0, 0) rectangle (8, -1) {};
\node at (3.5,-0.5) {$x_{1,1}$};
\draw[rounded corners] (9, 0) rectangle (17, -1) {};
\node at (12.5,-0.5) {$x_{1,2}$};

\draw[blue, rounded corners]  (-3,0.2) rectangle (17.2,-1.2) {};
\node at (-1.7,-0.5) {$P_v(1)$};

\draw[rounded corners] (0, -2) rectangle (8, -3) {};
\node at (3.5,-2.5) {$\neg x_{1,1}$};
\draw[rounded corners] (9, -2) rectangle (17, -3) {};
\node at (12.5,-2.5) {$\neg x_{1,2}$};

\draw[rounded corners] (0, -4) rectangle (8, -5) {};
\node at (3.5,-4.5) {$x_{2,1}$};
\draw[rounded corners] (9, -4) rectangle (17, -5) {};
\node at (12.5,-4.5) {$x_{2,2}$};

\draw[blue, rounded corners]  (-3,-3.8) rectangle (17.2,-5.2) {};
\node at (-1.7,-4.5) {$P_v(2)$};

\draw[rounded corners] (0, -6) rectangle (8, -7) {};
\node at (3.5,-6.5) {$\neg x_{2,1}$};

\draw[rounded corners] (9, -6) rectangle (17, -7) {};
\node at (12.5,-6.5) {$\neg x_{2,2}$};

\draw[rounded corners] (0, -8) rectangle (8, -9) {};
\node at (3.5,-8.5) {$x_{3,1}$};

\draw[blue, rounded corners]  (-3,-7.8) rectangle (17.2,-9.2) {};
\node at (-1.7,-8.5) {$P_v(3)$};

\draw[rounded corners] (0, -10) rectangle (8, -11) {};
\node at (3.5,-10.5) {$\neg x_{3,1}$};

\draw[rounded corners] (9, -12) rectangle (17, -13) {};
\node at (12.5,-12.5) {$x_{4,2}$};

\draw[rounded corners] (9, -14) rectangle (17, -15) {};
\node at (12.5,-14.5) {$\neg x_{4,2}$};

\draw[rounded corners] (0, -16) rectangle (8, -17) {};
\node at (3.5,-16.5) {$c_1$};
\draw[rounded corners] (9, -16) rectangle (17, -17) {};
\node at (12.5,-16.5) {$c_2$};

\filldraw (6.6,-2.5) circle(0.2);
\filldraw (7.5,-2.5) circle(0.2);

\filldraw (0.5,-0.5) circle(0.2);
\filldraw (0.5,-4.5) circle(0.2);
\filldraw (0.5,-8.5) circle(0.2);
\filldraw (0.5,-16.5) circle(0.2);

\filldraw (1.4,-0.5) circle(0.2);
\filldraw (1.4,-4.5) circle(0.2);
\filldraw (1.4,-10.5) circle(0.2);
\filldraw (1.4,-16.5) circle(0.2);

\filldraw (2.3,-0.5) circle(0.2);
\filldraw (2.3,-6.5) circle(0.2);
\filldraw (2.3,-8.5) circle(0.2);
\filldraw (2.3,-16.5) circle(0.2);

\filldraw (4.8,-0.5) circle(0.2);
\filldraw (4.8,-6.5) circle(0.2);
\filldraw (4.8,-10.5) circle(0.2);
\filldraw (4.8,-16.5) circle(0.2);

\filldraw (5.7,-2.5) circle(0.2);
\filldraw (5.7,-4.5) circle(0.2);
\filldraw (5.7,-8.5) circle(0.2);
\filldraw (5.7,-16.5) circle(0.2);

\filldraw (6.6,-2.5) circle(0.2);
\filldraw (6.6,-4.5) circle(0.2);
\filldraw (6.6,-10.5) circle(0.2);
\filldraw (6.6,-16.5) circle(0.2);

\filldraw (7.5,-2.5) circle(0.2);
\filldraw (7.5,-6.5) circle(0.2);
\filldraw (7.5,-10.5) circle(0.2);
\filldraw (7.5,-16.5) circle(0.2);

\filldraw (9.5,-0.5) circle(0.2);
\filldraw (9.5,-4.5) circle(0.2);
\filldraw (9.5,-12.5) circle(0.2);
\filldraw (9.5,-16.5) circle(0.2);

\filldraw (10.4,-0.5) circle(0.2);
\filldraw (10.4,-4.5) circle(0.2);
\filldraw (10.4,-14.5) circle(0.2);
\filldraw (10.4,-16.5) circle(0.2);

\filldraw (11.3,-0.5) circle(0.2);
\filldraw (11.3,-6.5) circle(0.2);
\filldraw (11.3,-12.5) circle(0.2);
\filldraw (11.3,-16.5) circle(0.2);

\filldraw (13.8,-2.5) circle(0.2);
\filldraw (13.8,-4.5) circle(0.2);
\filldraw (13.8,-12.5) circle(0.2);
\filldraw (13.8,-16.5) circle(0.2);

\filldraw (14.7,-2.5) circle(0.2);
\filldraw (14.7,-4.5) circle(0.2);
\filldraw (14.7,-14.5) circle(0.2);
\filldraw (14.7,-16.5) circle(0.2);

\filldraw (15.6,-2.5) circle(0.2);
\filldraw (15.6,-6.5) circle(0.2);
\filldraw (15.6,-12.5) circle(0.2);
\filldraw (15.6,-16.5) circle(0.2);

\filldraw (16.5,-2.5) circle(0.2);
\filldraw (16.5,-6.5) circle(0.2);
\filldraw (16.5,-14.5) circle(0.2);
\filldraw (16.5,-16.5) circle(0.2);
     \end{tikzpicture}
     \begin{tikzpicture}[scale=0.32]
\node at (7.5,-18.5) {Post-modification};
     
\draw[rounded corners] (0, 0) rectangle (8, -1) {};
\node at (3.5,-0.5) {$x_{1,1}$};
\draw[rounded corners] (9, 0) rectangle (17, -1) {};
\draw[red, rounded corners] (8.8, 0.2) rectangle (17.2, -1.2) {};
\node at (12.5,-0.5) {$x_{1,2}$};

\draw[rounded corners] (0, -2) rectangle (8, -3) {};
\node at (3.5,-2.5) {$\neg x_{1,1}$};
\draw[rounded corners] (9, -2) rectangle (17, -3) {};
\node at (12.5,-2.5) {$\neg x_{1,2}$};

\draw[rounded corners] (0, -4) rectangle (8, -5) {};
\node at (3.5,-4.5) {$x_{2,1}$};
\draw[rounded corners] (9, -4) rectangle (17, -5) {};
\node at (12.5,-4.5) {$x_{2,2}$};

\draw[blue, rounded corners]  (-3,-3.8) rectangle (17.2,-5.2) {};
\node at (-1.7,-4.5) {$P_v(2)$};

\draw[rounded corners] (0, -6) rectangle (8, -7) {};
\node at (3.5,-6.5) {$\neg x_{2,1}$};
\draw[rounded corners] (9, -6) rectangle (17, -7) {};
\node at (12.5,-6.5) {$\neg x_{2,2}$};

\draw[rounded corners] (0, -8) rectangle (8, -9) {};
\node at (3.5,-8.5) {$x_{3,1}$};

\draw[blue, rounded corners]  (-3,-7.8) rectangle (17.2,-9.2) {};
\node at (-1.7,-8.5) {$P_v(3)$};

\draw[rounded corners] (0, -10) rectangle (8, -11) {};
\node at (3.5,-10.5) {$\neg x_{3,1}$};

\draw[rounded corners] (9, -12) rectangle (17, -13) {};
\node at (12.5,-12.5) {$x_{4,2}$};

\draw[rounded corners] (9, -14) rectangle (17, -15) {};
\node at (12.5,-14.5) {$\neg x_{4,2}$};

\draw[rounded corners] (0, -16) rectangle (8, -17) {};
\node at (3.5,-16.5) {$c_1$};
\draw[rounded corners] (9, -16) rectangle (17, -17) {};
\node at (12.5,-16.5) {$c_2$};

\filldraw (6.6,-2.5) circle(0.2);
\filldraw (7.5,-2.5) circle(0.2);

\filldraw (0.5,-0.5) circle(0.2);
\filldraw (0.5,-4.5) circle(0.2);
\filldraw (0.5,-8.5) circle(0.2);
\filldraw (0.5,-16.5) circle(0.2);

\filldraw (1.4,-0.5) circle(0.2);
\filldraw (1.4,-4.5) circle(0.2);
\filldraw (1.4,-10.5) circle(0.2);
\filldraw (1.4,-16.5) circle(0.2);

\filldraw (2.3,-0.5) circle(0.2);
\filldraw (2.3,-6.5) circle(0.2);
\filldraw (2.3,-8.5) circle(0.2);
\filldraw (2.3,-16.5) circle(0.2);

\filldraw (4.8,-0.5) circle(0.2);
\filldraw (4.8,-6.5) circle(0.2);
\filldraw (4.8,-10.5) circle(0.2);
\filldraw (4.8,-16.5) circle(0.2);
\draw (4.8,-0.5)--(4.8,-16.5);

\filldraw (5.7,-2.5) circle(0.2);
\filldraw (5.7,-4.5) circle(0.2);
\filldraw (5.7,-8.5) circle(0.2);
\filldraw (5.7,-16.5) circle(0.2);

\filldraw (6.6,-2.5) circle(0.2);
\filldraw (6.6,-4.5) circle(0.2);
\filldraw (6.6,-10.5) circle(0.2);
\filldraw (6.6,-16.5) circle(0.2);

\filldraw (7.5,-2.5) circle(0.2);
\filldraw (7.5,-6.5) circle(0.2);
\filldraw (7.5,-10.5) circle(0.2);
\filldraw (7.5,-16.5) circle(0.2);

\filldraw (9.5,-0.5) circle(0.2);
\filldraw (9.5,-4.5) circle(0.2);
\filldraw (9.5,-12.5) circle(0.2);
\filldraw (9.5,-16.5) circle(0.2);

\filldraw (10.4,-0.5) circle(0.2);
\filldraw (10.4,-4.5) circle(0.2);
\filldraw (10.4,-14.5) circle(0.2);
\filldraw (10.4,-16.5) circle(0.2);

\filldraw (11.3,-0.5) circle(0.2);
\filldraw (11.3,-6.5) circle(0.2);
\filldraw (11.3,-12.5) circle(0.2);
\filldraw (11.3,-16.5) circle(0.2);

\filldraw (13.8,-2.5) circle(0.2);
\filldraw (13.8,-4.5) circle(0.2);
\filldraw (13.8,-12.5) circle(0.2);
\filldraw (13.8,-16.5) circle(0.2);

\filldraw (14.7,-2.5) circle(0.2);
\filldraw (14.7,-4.5) circle(0.2);
\filldraw (14.7,-14.5) circle(0.2);
\filldraw (14.7,-16.5) circle(0.2);

\filldraw (15.6,-2.5) circle(0.2);
\filldraw (15.6,-6.5) circle(0.2);
\filldraw (15.6,-12.5) circle(0.2);
\filldraw (15.6,-16.5) circle(0.2);

\filldraw (16.5,-2.5) circle(0.2);
\filldraw (16.5,-6.5) circle(0.2);
\filldraw (16.5,-14.5) circle(0.2);
\filldraw (16.5,-16.5) circle(0.2);
     \end{tikzpicture}   
     \caption{Illustration of re-pooling in Lemma~\ref{lem:ccont} proof for $n=4$, $m=2$, demonstrated on the clause $x_1 \lor x_2 \lor x_3$ with the index $j=1$. The black rectangles represent states, the blue rectangles -- attractive variable pools, the black lines -- clause pools (with dots in the states belonging to them) and the red rectangle -- a singleton pool. Under the transformation, the attractive variable pool $P_v\parentheses*{1}$ is removed, and a new clause pool $\braces*{c_1 , x_{1,1}, \neg x_{2,1}, \neg x_{3,1}}$ is created, together with a new singleton pool $\braces*{x_{1,2}}$. The signaling policy in states $\neg x_{1,1}, \neg x_{1,2}, \neg x_{2,2}, x_{4,2},\neg x_{4,2}$ and $c_2$ remains unchanged during the modification.} 
     \label{fig:ccont}
 \end{figure}

Since in all three states $x_{i_1,j}$, $x_{i_2,j}$ and $x_{i_3,j}$ the sent signal is the one that induces the variable pool posterior, one cannot use these states to pool $c_j$ together with them. Therefore, at state $c_j$, the sender's utility must be $-7$. We apply the following modification to the equilibrium (see Figure~\ref{fig:ccont} for illustration). We break the pool $P_v\parentheses*{i_1}$. Now we pool the states $\braces*{c_j , x_{i_1,j}, \neg x_{i_2,j}, \neg x_{i_3,j}}$ into a clause pool. Finally, for all the now un-pooled states $x_{i,j'}$ ($j'\neq j$), we pool them into singleton pools. The resultant sender's strategy still yields a cheap talk equilibrium, assuming that the receiver takes the corresponding action $a_p$ for each newly created pool $p$.

The sender's utility gain from this modification is $7$ for state $c_j$, whereas her losses from this modification appear in the $d$ states of $P_v\parentheses*{i_1}$ in which the utility has been reduced from $1$ to $0$. Since we assume $d\leq 6$, such a modification strictly increases the sender's utility, as desired. 

\section{Omitted Proofs in Section \ref{sec:tractble}} 
\subsection{Proof of Proposition~\ref{pro:binary}}\label{append:tractable:pro-binary}
By Proposition~\ref{pro:signals}, one can restrict attention to cheap talk equilibria in which the sender uses finitely many signals. Assume w.l.o.g.~that there is no state $\omega\in\Omega$ with $u_S\parentheses*{\omega,a_1}=u_S\parentheses*{\omega,a_2}$ and $u_R\parentheses*{\omega,a_1}=u_R\parentheses*{\omega,a_2}$. Indeed, the utility of both agents in such a state is independent of $\pi$ and $s$, implying that such a state can be ignored in the equilibrium analysis. Let $\Omega^1$, $\Omega^2$ and $\Omega^{1,2}$ be, respectively, the sets of states in which the sender strictly prefers action $a_1$ over action $a_2$, strictly prefers $a_2$ over $a_1$, and is indifferent between the two actions.

Consider a sender-optimal equilibrium $\parentheses*{\pi, s}$ with $\supp\parentheses*{\pi}$ being the smallest among all possibilities for a sender-optimal equilibrium. Note that no two signals $\sigma^1,\sigma^2\in\supp\parentheses*{\pi}$ might have $s\brackets*{\sigma^1}=s\brackets*{\sigma^2}$, as replacing them by a single signal leads to an equivalent equilibrium in terms of utilities, contradicting the minimality of $\supp\parentheses*{\pi}$. Indeed, given two such signals $\sigma^1$ and $\sigma^2$ -- one can define a new strategy profile $\parentheses*{\pi', s'}$ s.t.: (i)~$\pi'$ is obtained from $\pi$ by replacing $\sigma^2$ with $\sigma^1$; (ii)~$s'=s|_{\supp\parentheses*{\pi'}}$ is obtained from $s$ by setting $s'\brackets*{\sigma}=s\brackets*{\sigma}$ for $\sigma\neq \sigma^2$. Since the posterior $p_{\sigma^1}^{\pi'}$ is a convex combination of the posteriors $p_{\sigma^1}^{\pi}$ and $p_{\sigma^2}^{\pi}$, we get that $s\brackets*{\sigma^1}$ is a best response to getting the signal $\sigma_1$ under $\pi'$.\footnote{\label{foot:1}Here we use the well-known fact that the set of posteriors for which a specific action is (weakly) best for the receiver is convex (see, e.g.,~\cite{kamenica2011bayesian}).} Moreover, as the sender's expected utility upon transmitting any signal under $\pi'$ is the same as under $\pi$, and the receiver's mixed strategy upon getting a specific signal under $\pi'$ is the same as upon getting this signal under $\pi$, we get that $\pi'$ is a best response to $s'$. To sum up, $\parentheses*{\pi', s'}$ is a cheap talk equilibrium that has the same expected sender's and receiver's utilities as $\parentheses*{\pi, s}$, but $\supp\parentheses*{\pi'}<\supp\parentheses*{\pi}$, a contradiction.

Let $\sigma_1\in\supp\parentheses*{\pi}$ be the signal upon getting which the receiver takes action $a_1$ with the highest probability over all signals in $\supp\parentheses*{\pi}$, say with probability $q_1$. Note that $\sigma_1$ is uniquely defined by the previous paragraph. Similarly, let $\sigma_2$ be the unique signal resulting in the receiver taking action $a_2$ with the highest probability among the signals in $\supp\parentheses*{\pi}$, say with probability $q_2$.

We claim that for any $\omega\in\Omega^1$, $\pi\parentheses*{\sigma_1\mid \omega}=1$. Indeed, if that condition is not satisfied for some $\omega_i\in\Omega^1$, then the sender will better off deviating from $\pi$ to $\tilde{\pi}$ with $\tilde{\pi}\brackets*{\omega}=\pi\brackets*{\omega}$ for $\omega\neq\omega_i$ and $\tilde{\pi}\parentheses*{\sigma_1\mid \omega_i}=1$, a contradiction. Thus, $\sigma_1$ is sent with probability $1$ for any $\omega\in\Omega^1$. Similarly, $\sigma_2$ is sent with probability $1$ for any $\omega\in\Omega^2$.

Consider some $\omega\in\Omega^{1,2}$ and let $\sigma_3$ be some signal sent with a positive probability in the state $\omega$. If $\sigma_3\neq \sigma_1, \sigma_2$, then $\sigma_3$ is sent with a positive probability exclusively for $\omega\in\Omega^{1,2}$. Note that it is w.l.o.g.~to assume $\absolute*{\supp\parentheses*{s\brackets*{\sigma_3}}}=1$, since if both agents are indifferent between both actions upon $\sigma_3$, then making the receiver choose action $a_2$ upon observing $\sigma_3$ does not break the equilibrium conditions, and the expected utilities of both agents are unchanged. Assume, w.l.o.g., that $\supp\parentheses*{s\brackets*{\sigma_3}}=\braces*{a_2}$. Then sending $\sigma_2$ whenever the sender should have transmitted $\sigma_3$ (under $\pi$) still yields a posterior for which $a_2$ is a best response for the receiver (see Footnote~\ref{foot:1}). Therefore, modifying $\pi$ by sending $\sigma_2$ whenever one should have sent $\sigma_3$ yields a sender-optimal equilibrium using a smaller number of signals, a contradiction. Therefore, $\sigma_1$ and $\sigma_2$ are the only signals used by the sender.

Moreover, let $\Omega^1_2:=\braces*{\omega\in\Omega^{1,2}:\; u_R\parentheses*{\omega,a_1}>u_R\parentheses*{\omega,a_2}}$ and $\Omega^2_1:=\Omega^{1,2}\setminus \Omega^1_2$. Then, by the discussion above, $\sigma_1$ is sent with probability $1$ for $\omega\in\Omega^1\cup\Omega^1_2$ and $\sigma_2$ is sent with probability $1$ for $\omega\in\Omega^2\cup\Omega^2_1$. Note that the expected sender's utility is linear in $q_1,q_2$ and non-decreasing in each of~them.

Recall that in each state, the signal sent is chosen deterministically. Therefore, we are only left with the receiver's incentive-compatibility constraints. Let us normalize the receiver's utility for action $a_2$ to be $0$ in any state.\footnote{It is possible since subtracting the same number from $u_R\parentheses*{\omega,a_1}$ and $u_R\parentheses*{\omega,a_2}$ for a fixed $\omega\in\Omega$ does not affect the set of equilibria and sender's preferences over them.} Note that if $q_1>0$ and $q_2>0$, we must have: $\sum_{\omega\in\Omega^1\cup\Omega^1_2} u_R\parentheses*{\omega,a_1}\cdot \mu\parentheses*{\omega}\geq 0$ and  $\sum_{\omega\in\Omega^2\cup\Omega^2_1} u_R\parentheses*{\omega,a_1}\cdot \mu\parentheses*{\omega}\leq 0$. If both these conditions hold, then in the sender-optimal equilibrium we have $q_1=q_2=1$, which corresponds to a sender-greedy signaling policy and the receiver obeying the recommended actions. If at least one of these conditions does not hold, then $q_1q_2=0$, implying that the receiver takes one action regardless of the signal. In this case, the sender cannot benefit from cheap talk, and the sender-optimal equilibrium is a babbling equilibrium using just one signal; if the receiver is indifferent between the two actions when the posterior equals the prior, then he should break the tie in the sender's favor.

Since there is no state in which both the sender and the receiver are indifferent between the two actions, there is only one possible sender-greedy signaling policy $\pi$. If the receiver chooses the strategy $s$ that obeys the sender's recommendations according to $\pi$, then from the conditions $\sum_{\omega\in\Omega^1\cup\Omega^1_2} u_R\parentheses*{\omega,a_1}\cdot \mu\parentheses*{\omega}\geq 0$ and  $\sum_{\omega\in\Omega^2\cup\Omega^2_1} u_R\parentheses*{\omega,a_1}\cdot \mu\parentheses*{\omega}\leq 0$ we get that the receiver best-replies to the sender. Moreover, the sender-greedy signaling policy best-replies to $s$ by definition. Therefore, the resultant $\parentheses*{\pi,s}$ is an equilibrium. Moreover, the sender's expected utility can be computed in time linear in $n$. As the sender's expected utility in the two babbling equilibria in which the receiver plays a pure strategy can also be computed in time linear in $n$, we are done.
\end{document}